\documentclass{LMCS}

\def\dOi{10(4:10)2014}
\lmcsheading%
{\dOi}
{1--41}
{}
{}
{Dec.~30, 2013}
{Dec.~16, 2014}
{}

\ACMCCS{[{\bf Theory of computation}]: Semantics and
  reasoning---Program semantics}

\subjclass{F.3.2 Semantics of Programming Languages}

\usepackage[utf8]{inputenc}
\usepackage{amssymb}
\usepackage{stmaryrd}
\usepackage{graphicx}
\usepackage{bussproofs}
\usepackage{hyperref}
\usepackage{xspace}

\mathcode`\<="4268      
\mathcode`\>="5269      
\mathchardef\gt="313E 
\mathchardef\lt="313C 

\theoremstyle{definition}
\newtheorem{definition}[thm]{Definition}
\newtheorem{example}[thm]{Example}
\theoremstyle{plain}
\newtheorem{lemma}[thm]{Lemma}
\newtheorem{proposition}[thm]{Proposition}
\newtheorem{theorem}[thm]{Theorem}

\newcommand{\red}{\longrightarrow}

\newcommand{\lollipop}{\to}
\newcommand{\lolli}[3]{#1\cdot #2 \to #3}

\newcommand{\kw}[1]{\mathsf{#1}}

\newcommand{\NN}{\mathbb{N}}
\newcommand{\unit}{\mathtt{unit}}
\newcommand{\VN}{\mathtt{nat}}

\newcommand{\length}[1]{\mathit{length}(#1)}

\newcommand{\Capply}{\textit{apply}}

\newcommand{\Mid}{\ \mathrel{\big|} \ }
\newcommand{\I}[2]{#1 \colon #2}
\newcommand{\J}[3]{#1 \colon #2 \cdot #3}
\newcommand{\R}[1]{\textsc{#1}}

\newcommand{\SeqTm}[3]{#1 \vdash #2 \colon #3}

\newcommand{\SeqW}[3]{#1 \,\vdash\, #2\,\, {:}\,\, #3}
\newcommand{\SeqU}[3]{#1 \,\vdash\, #2\,\, {:}\,\, #3}

\newcommand{\LL}{\mathcal{L}}
\newcommand{\LT}{\mathcal{L}_T}
\newcommand{\LeftLabelSc}[1]{\LeftLabel{\textsc{#1}}}
\newcommand{\sleq}{\triangleleft}

\newcommand{\xto}{\xrightarrow}

\newcommand{\inl}{\kw{inl}}
\newcommand{\inr}{\kw{inr}}

\newcommand{\tlam}[2]{\lambda {#1}.\, {#2}}
\newcommand{\tlaml}[4]{\lambda^{#3} {#1}{:}{#2}.\, {#4}}
\newcommand{\tlamlu}[4]{\lambda^{#3} {#1}.\, {#4}}
\newcommand{\tlami}[3]{\lambda {#1}{:}{#2}.\, {#3}}
\newcommand{\tappl}[3]{{#1} {\,@_{#2}\,} {#3}}
\newcommand{\tlet}[3]{\kw{let}\ #2=#1\ \kw{in}\ #3}
\newcommand{\tif}[3]{\kw{if0}\ #1\ \kw{then}\ #2\ \kw{else}\ #3}
\newcommand{\tcase}[5]{\kw{case}\ #1\ \kw{of}\ \kw{inl}(#2) \Rightarrow #3;\, \kw{inr}(#4) \Rightarrow #5}
\newcommand{\tmcase}[5]{%
  \!\begin{aligned}[t]
  \kw{case}\ #1\ \kw{of}\ 
     \kw{inl}(#2) &\Rightarrow #3 \\[-.27em]
    ;\, \kw{inr}(#4) &\Rightarrow #5    
  \end{aligned}
  }
\newcommand{\fcase}[5]{\kw{case}\ #1\ \kw{of}\ #2 \Rightarrow #3 ;\, #4 \Rightarrow #5}

\newcommand{\tlist}[1]{#1\ \kw{list}}

\newcommand{\cps}[1]{\underline{#1}}

\newcommand{\TT}{\mathbb{T}}
\newcommand{\Int}[1]{\textrm{Int}(#1)}

\newcommand{\sem}[1]{\llbracket #1 \rrbracket}
\newcommand{\semc}[1]{\kw{Int}(#1)}

\newcommand{\iso}{\simeq}

\newcommand{\linexp}{\textsc{lin}$_\textsc{exp}$\xspace}
\newcommand{\stlexp}{\textsc{stl}$_\textsc{exp}$\xspace}

\begin{document}

\title[Interaction Semantics, Continuations and
  Defunctionalization] {On the Relation of Interaction Semantics to
  Continuations and Defunctionalization}

\author[U.~Schöpp]{Ulrich Schöpp}
\address{Ludwig-Maximilians-Universität München, Germany}
\email{Ulrich.Schoepp@ifi.lmu.de}

\keywords{CPS-translation, Defunctionalization,
Int-construction, Game Semantics, Geometry of Interaction}

\begin{abstract}
  \noindent
  In game semantics and related approaches to programming 
  language semantics, programs are modelled by interaction
  dialogues.  Such models have recently been used in the design of new compilation methods,  
  e.g.~for hardware synthesis or for programming with sublinear space.
  This paper relates such semantically motivated non-standard
  compilation methods to more standard techniques in the compilation
  of functional programming languages, namely continuation passing and
  defunctionalization. We first show for the linear $\lambda$-calculus that
  interpretation in a model of computation by interaction can be
  described as a call-by-name CPS-translation followed by a
  defunctionalization procedure that takes into account control-flow
  information. 
  We then establish a 
  relation between these two compilation methods for the 
  simply-typed $\lambda$-calculus and end by considering recursion.
\end{abstract}

\maketitle

\section{Introduction}

A successful approach in the semantics of programming languages is to model
programs using interaction dialogues.
It is fundamental to Game Semantics~\cite{HylandOng,AbramskyJM00},
the Geometry of Interaction~\cite{Girard89} and
related lines of research.
The idea goes back to the study of dialogical models of 
constructive logic~\cite{lorenzen1961dk},
which explain the meaning of a logical sentence 
by how one can attack and defend it in a 
debate~\cite{DBLP:journals/apal/Blass92}. A proof of a
sentence is a strategy for defending it against any possible attack.  In
programming language semantics, types take the place of sentences and attacks
can be seen as requests for information. The meaning of a program is
a strategy that explains how to answer any possible request.
Programs are interpreted compositionally, so that the answer
to a request depends only on how the parts of the programs answer
to suitable requests. Computation is thus modelled as an interaction
dialogue.

While interaction dialogues are typically considered as abstract mathematical
objects, it has also been argued that they are useful for \emph{implementing}
actual computation.  To compute the result of a program it is enough to have an
implementation of the strategy that interprets it, i.e.~a implementation that takes
requests as input and that computes the strategies' answer as output.  The
compositional definition of the interactive interpretation guides the
construction of such an implementation.  For example, one may implement the
strategy for each program part by a separate module.  The compositional
translation of programs explain how to assemble such modules to obtain the
implementation of a whole program.  The modules interact with each other by a suitable
form of message passing and implement the computation by playing out actual
interaction dialogues.
Implementations of this kind have been proposed 
for example in~\cite{DBLP:conf/tgc/FredrikssonG12,DBLP:conf/lics/FredrikssonG13,intml}.

One main motivation for studying the implementation of interaction models is 
to guide the design of compilation methods for programming languages. Interaction models are typically quite
concrete and suitable for implementation in simple low-level languages,
but, at the same time, they have rich structure and provide accurate models for
sophisticated programming languages, see e.g.~\cite{AbramskyJM00,HylandOng,DBLP:journals/apal/MurawskiT13}. 

The approach of using interactive semantics as an implementation technique for
programming languages has been  proposed in a variety of contexts. 
Mackie~\cite{DBLP:conf/popl/Mackie95} uses ideas from
the Geometry of Interaction for the implementation of functional
languages.
In later work it was noticed that such ideas are useful especially for
the implementation of functional languages with strong resource constraints.
Ghica et al. have developed methods for hardware synthesis
based on Game Semantics~\cite{Ghica07,DBLP:conf/icfp/GhicaSS11}. A 
related semantic approach has been
used to design a functional programming language for sublinear
space computation~\cite{intml}.
Other work has been motivated by the idea that strategies are implemented 
by communicating modules. 
This has inspired work on fully abstract translations from PCF
to the $\pi$-calculus, such as~\cite{DBLP:conf/fpca/HylandO95,BergerHY01}.
It has also been used to apply ideas from Game Semantics and the Geometry
of Interaction for distributed computing~\cite{DBLP:conf/tgc/FredrikssonG12}.
In another direction, the Geometry of Interaction is being used 
as a basis for structuring quantum computation~\cite{DBLP:conf/lics/HasuoH11,DBLP:conf/esop/YoshimizuHFL14}.
This list of examples is certainly not exhaustive; it illustrates the 
wide range of applications of the implementation of interactive dialogues.

The aim of this paper is to relate compilation methods based on 
interaction semantics to standard techniques in the
efficient compilation of functional programming languages.
It has been observed before, for example by Melli\`es~\cite{mellies12} and
Levy~\cite{cbpv}, that interaction models are related to
continuation passing, an important standard technique in the compilation
of functional programming languages~\cite{DBLP:books/cu/Appel1992}.
In this paper we make a further connection to
defunctionalization~\cite{reynolds}.

We consider the compilation of higher-order languages, such as PCF. 
A compiler would transform such a 
language to machine code by way of a number of intermediate languages. 
Typically, the higher-order source code would first be translated to 
first-order intermediate code,
from which the machine code is then generated. 
This paper is concerned with the first step,
the translation from higher-order to first-order code.
We show that the composition of two well-known transformations,
namely CPS-translation~\cite{plotkin} and defunctionalization~\cite{reynolds}, 
is closely related to an interpretation 
of the source language in a model that implements interaction dialogues.

The interactive model we study in this paper is an instance of 
the \emph{Int construction}~\cite{Joyal96}.
This model is very basic and captures only what is needed for its
 intended application as an implementation technique. 
We believe that it is a good choice, as the Int construction has been 
identified as the core of a number of interactive semantics, 
so that our results apply to a number of interactive models.
Indeed, in~\cite{DBLP:journals/mscs/AbramskyHS02} it was shown that 
the (particle style) Geometry of
Interaction can be seen as an instance of the Int construction with further
structure.  Abramsky-Jagadeesan-Malacaria (AJM) games~\cite{AbramskyJM00}
are also closely related to the Int construction.
AJM games refine the Int construction by removing 
unwanted interaction dialogues and by integrating a quotient to capture a good
notion of program equality, see the construction in~\cite{AbramskyJM00}.  
If one is interested only in implementing
strategies, then one may restrict ones attention to the core given only by the
Int construction.

In order to define an interpretation of a higher-order source language in an
interactive model given by the Int construction, we build on work reported in~\cite{intml}.
As the resulting interpretation implements call-by-name, we
relate it to a call-by-name CPS-translation -- a variant of the
one by Hofmann and Streicher~\cite{DBLP:conf/lics/HofmannS97}.

Let us outline concretely how CPS-translation, defunctionalization
and the interpretation in an interactive model are related
by looking at the very simple example of a function that increments a natural number:
$\tlami x \NN {1 + x}$.
We next outline how this function is translated by the two
approaches and how the results compare.

\subsection{CPS-Translation and Defunctionalization}

A compiler for PCF might first transform $\tlami x \NN {1 + x}$
into continuation passing style, perhaps apply some optimisations, and then use
defunctionalization to obtain a first-order intermediate program,
ready for compilation to machine language.

Hofmann and Streicher's call-by-name CPS-translation~\cite{DBLP:conf/lics/HofmannS97} 
translates the source term
$\tlami x \NN {1 + x}$ to 
$\lambda <x,k>.\, (\lambda k. k\ 1)\ (\lambda u.\, x\ (\lambda n.\, k\ (u+n)))
\colon \neg(\neg \neg \NN \times \neg \NN)$,
where we write $\neg A$ for $A\to \bot$.
This term defines a function, which takes as argument a pair $<x,k>$ of a continuation
$k\colon \neg \NN$ that accepts the result and
a variable $x\colon \neg\neg\NN$ that supplies the 
function argument. To obtain the actual function argument, one applies~$x$ to
a continuation (here $\lambda n.\, k\ (u+n)$) to 
ask for the actual argument to be thrown into the supplied 
continuation.

Defunctionalization~\cite{reynolds} translates this higher-order term into a
first-order program. The basic idea is to give each function a name 
and to pass around not the function itself, but only its name 
and the values of its free variables. 
To this end, each $\lambda$-abstraction is named with a unique label:
$\lambda^{l_1} <x,k>.\, (\lambda^{l_2} k. k\ 1)\ (\lambda^{l_3} u.\, x\ (\lambda^{l_4} n.\, k\ (u+n)))$.
The whole term defines the function named with label~$l_1$. 
It can be represented simply by the label~$l_1$. The 
function with label~$l_3$ has free variables~$x$ and~$k$ and 
is represented by the label 
together with the values of~$x$ and~$k$, which we write as~$l_3(x,k)$.

Each  application $s\ t$ is replaced by a 
procedure call $\Capply(s,t)$, as~$s$ is now only the name of a function and not a
function itself. The procedure~$\Capply$ is defined by case distinction on
the function name and behaves like the body of the respective
$\lambda$-abstraction in the original term. In the example,
we have the following definition of $\Capply$:
\[
\Capply(f, a) =
\begin{aligned}[t]
  \kw{case}\ f\ \kw{of}\
   &l_1\Rightarrow \tlet a {<x,k>} {\Capply(l_2, l_3(x,k))}\\
   \mid\ &l_2 \Rightarrow \Capply(a, 1)\\
   \mid\ &l_3(x,k) \Rightarrow \Capply(x, l_4(k,a))\\
   \mid\ &l_4(k,u) \Rightarrow \Capply(k, u + a)
\end{aligned}
\]
This definition should be understood as the recursive definition of 
a function $\Capply$ with two arguments. The definition is untyped, as in Reynold's original
definition of defunctionalization~\cite{reynolds}.

To understand concretely how this definition represents the original term, 
it is perhaps useful to see what happens when
a concrete argument and a continuation are supplied:
$ (\lambda^{l_1} <x,k>.\, (\lambda^{l_2} k. k\ 1)\ (\lambda^{l_3} u.\, x\ (\lambda^{l_4} n.\, k\ (u+n))))
\ <\lambda^{l_5} k.\, k\ 42,\, \lambda^{l_6} n.\, \texttt{print\_int}(n)>.
$
The definition of $\Capply$ then has two cases for $l_5$ and $l_6$
in addition to the cases above:
\[
\Capply(l, a) =
\begin{aligned}[t]
\kw{case}\ l\ \kw{of}\ 
   &\dots \\
   \mid\ &l_5 \Rightarrow \Capply(a, 42)\\
   \mid\ &l_6 \Rightarrow \texttt{print\_int}(n)
\end{aligned}
\]
The fully applied term defunctionalizes to $\Capply(l_1,<l_5,l_6>)$.
Executing it results in~$43$ being printed. 
When we evaluate $\Capply(l_1,<l_5,l_6>)$, the first case in the definition
of $\Capply$ applies and results in the call $\Capply(l_2,l_3(l_5,l_6))$.
For this call, the second case applies, so that the call 
$\Capply(l_3(l_5,l_6),1)$ is made. The computation continues in
this way with calls to
$\Capply(l_5, l_4(l_6,1))$,
$\Capply(l_4(l_6,1), 42)$,
$\Capply(l_6, 43)$, and finally
$\texttt{print\_int}(43)$.

This outlines a naive defunctionalization method for translating a
higher-order language into a first-order language with (tail) recursion. 
This method can be improved in various ways. The above
$\Capply$-function performs a case distinction on the function name
each time it is invoked.  However, in the example
it is possible to determine the label in the first argument of each
appearance of $\Capply$ statically, so that the 
case distinction is not necessary.
Instead, we may define one function $\Capply_l$ for each label~$l$
and replace $\Capply(l(x),a)$ by $\Capply_l(x,a)$.
The label~$l$ thus does not need to be passed as an argument anymore. A 
defunctionalization procedure that takes into account control flow
information in this way was introduced by Banerjee et
al.~\cite{banerjee}.
If we apply it to this example and moreover simplify the result by removing
unneeded function arguments, then we get four mutually recursive functions:
\begin{equation}
\begin{aligned}
  \label{eq:defun}
  \Capply_{l_1}() &= \Capply_{l_2}() &
  \Capply_{l_2}() &= \Capply_{l_3}(1) \\ 
  \Capply_{l_3}(u) &= \Capply_{l_5}(u) &
  \hspace{1.9cm}\Capply_{l_4}(u, n) &= \Capply_{l_6}(u+n) \\
\end{aligned}
\end{equation}
The term itself simplifies to $\Capply_{l_1}()$.
The interface where these equations interact with the environment 
consists of the labels $l_1$ (the entry label), $l_6$ (the return label), 
$l_5$ (the entry label for argument function~$x$) and $l_4$ (the return label for
the argument function~$x$). 
Applying the term to concrete arguments as above 
amounts to extending the environment with the following equations:
\begin{align*}
  \Capply_{l_5}(u) &= \Capply_{l_4}(u,42) 
  & \Capply_{l_6}(n) &= \texttt{print\_int}(n)
\end{align*}

The point of this paper is that the program (\ref{eq:defun}) is just  
what we get from interpreting the source term in a model of computation by
interaction.

\subsection{Interpretation in an Interactive Computation Model}
\label{sect:particle}

In computation by interaction the general idea is to study
models of computation that interpret programs by interaction dialogues 
and to consider actual
implementations of such dialogue interaction. 
For example, a function of type $\NN\to \NN$ may be implemented
in interactive style by a program that, for a suitable type~$S$, takes
as input a value of type $\unit + (S\times \VN)$ and gives as output a value
of type $\VN + (S\times \unit)$. The input $\inl(<>)$ to this program is
interpreted as a request for the return value of the function.  An
output of the form~$\inl(n)$ means that~$n$ is the requested value.  If
the output is of the form~$\inr(s, <>)$, however, then this means that
the program would like to know the argument of the function. 
It also requests that the value~$s$ be returned along with the answer.
Programs here do not have state and 
have no persistent memory to store any data until 
a request is answered.
The program can however encode any data that it needs later in the value~$s$
and ask for this value to be returned unchanged with the answer to its
request. 
For the outside, the value~$s$ is opaque. We do not know anything about 
what is encoded in the value~$s$, only that we have to give it back with the 
answer to the request. 
To answer the program's request, we pass
a value of the form $\inr(s,m)$, where~$m$ is our answer.

The particular function $\tlami x \NN {1 + x}$ 
is implemented by the program specified in the following diagram,
where $S$ is $\VN$.
This diagram is to be understood so that one may pass a message along
any of its input wires. The message must be a value of the type labelling
the wire. When a message arrives at an input of box, the box will react 
by sending a message on one of its outputs. Thus, at any time there is
one message in the network. Computation ends when a message is passed
along an output wire. 
\begin{center}
  \includegraphics{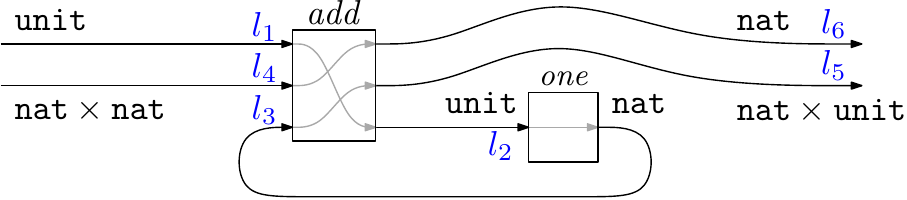}
\end{center}
In this diagram, 
$\textit{add}$ has three input ports of type $\unit$, $\VN\times \VN$ and $\VN$
respectively from top to bottom.
It may output a message on one of its three output ports, which
have type $\VN$, $\VN \times \unit$ and $\unit$ from top to bottom.
Its behaviour is given as follows:
if it receives message $<>$ on the topmost input port (a request for the sum), 
then it outputs $<>$ on the bottom output port (a request to provide the first
summand); 
if it receives $n$ on the bottom input port (the first summand),
then it outputs $<n,<>>$ on the middle output port (a request to provide the second
summand and to hold on to the first summand until the request is answered); and if
it receives $<n,m>$ on the middle input port (both summands), it outputs $n+m$ 
on the topmost output port.
The box labelled~$\textit{one}$ 
maps the 
request~$<>$ to the number~$1$.

This interactive implementation of $\tlami x \NN {1 + x}$ may be
described as the interpretation of the term in a semantic
model $\Int\TT$ built by applying the general categorical Int construction 
to a category~$\TT$ that is constructed from the target language,
see Section~\ref{sect:linear}.

Compare the above interaction diagram to the definitions in~(\ref{eq:defun}) obtained 
by defunctionalization.
The labels $l_1$, $l_4$ and $l_3$ there correspond to the
three input ports of the $\textit{add}$-box (from top to bottom), $l_2$ is
the input of box~$\textit{one}$, and $l_5$ and $l_6$ are the
destination labels of the two outgoing wires. 
One may consider the $\Capply$-definitions in~(\ref{eq:defun}) a particular 
implementation of the diagram, where a call to $\Capply_l(m)$ means 
that message~$m$ is sent to point~$l$ in the diagram.
A naive implementation would introduce a label for the end of each
arrow in the diagram and implement the message passing accordingly.

\subsection{Overview}

The subject of this paper is the relation of the two translations that
we have just outlined.
The paper studies the following situation of two
translations from a source language that is a variant of 
PCF into a simple first-order target language with tail recursion.
\medskip
\begin{center}
\includegraphics{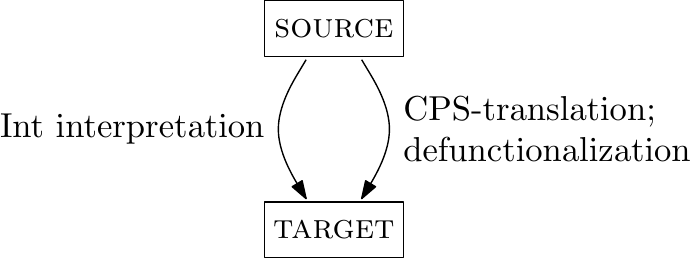}
\end{center}
\medskip

After giving definitions of the source and target language in the following
two sections, we find it useful to present and analyse the above situation step by step.
We define the two translations and study their relationship for a number of 
fragments of the source language of increasing strength.
\[
\textsc{core}
\ \subseteq\
\textsc{lin}
\ \subseteq\
\textsc{stl}
\ \subseteq\
\textsc{source}
\]

In Section~\ref{sect:linear} we start by studying the translation for the
source fragment \textsc{core}.
This fragment contains just a bare minimum of 
linear $\lambda$-abstraction and application. 
It is nevertheless instructive to consider this fragment as
a setting in which to develop the infrastructure for the translation of
higher-order functions. 

In Section~\ref{sect:base} we consider the fragment \textsc{lin}, which
extends \textsc{core} with a base type of natural numbers. 
Rather than studying the above situation directly with \textsc{lin} in
place of \textsc{source}, we argue that it is useful to take a detour
over a calculus \linexp, which is a version of \textsc{lin} with additional 
type annotations. 
These type annotations are useful for understanding the Int-interpretation.

In Section~\ref{sect:stl} we then add contraction and come to
the simply typed fragment \textsc{stl} of the source language.
We first continue to use additional type annotations and extend 
\linexp to \stlexp. We then come back to the unannotated source
fragment \textsc{stl} by showing how \textsc{stl} can be translated
into \stlexp (Prop.~\ref{prop:stl}). 

In Section~\ref{sect:fix} we finally then extend the translation to the
full source language by adding recursion.

\section{Target Language}
\label{sect:target}

Programs in the target language consist of mutually tail-recursive definitions
of first-order functions, such as the $\Capply$-equations above. 
One should think of the target language as a simple variant
of SSA-form compiler intermediate 
languages, e.g.~\cite{DBLP:conf/esop/CejtinJW00}, in which
function definitions are often presented as labelled blocks that
end with a jump to a label.

The target language does not model function calls or a calling
convention; 
it models only what one would use for the compilation of a
single unit.
Certain function labels are designated as entry or exit
points. In the following example target program
the labels $\textit{const}$ and $\textit{pow}$ are
intended as entry points.
\[
\begin{aligned}
  {\textit{const}}(x) &= {\textit{const\_ret}}(23) \\
  {\textit{pow}}(<x,y>) &= {\textit{pow\_loop}}(<x,y>) \\
  {\textit{pow\_loop}}(<x,y>) &= \tmcase {\kw{iszero}(x)} {z} {{\textit{pow\_ret}}(y)} {z} {{\textit{pow\_loop}}(<x-1,y*y>)}
\end{aligned}
\]
The function labels
$\textit{const\_ret}$ and $\textit{pow\_ret}$ are exit points that
are assumed to be defined externally and that are used to return the 
results of computations. 

A target program will be a set of equations
together with lists of entry and exit labels that specify the
interface of the program. Target programs are defined in detail in the
rest of this section. Upon first reading, the reader may wish to skim
this section only.

Target programs are typed.
The set of \emph{target types} is defined by the grammar below.
Recursive types will be needed at the end of Section~\ref{sect:stl} only.
\emph{Target expressions} are standard terms for these types, see
e.g.~\cite{pierce}:
\begin{align*}
  \text{Types: } && 
   A, B &::= \alpha \Mid \unit \Mid \VN \Mid A \times B \Mid A + B \Mid \mu \alpha.\, A\\
  \text{Expressions: } && e,e_1,e_2 &::= 
    \begin{aligned}[t]
      x &\Mid <>
      \Mid n \Mid e_1 + e_2 \Mid \kw{iszero}(e) \\
      &
      \Mid <e_1, e_2> \Mid \tlet {e_1} {<x, y>} {e_2}\\
      &
      \Mid \kw{inl}(e) \Mid \kw{inr}(e) 
      \Mid \tcase {e} x {e_1} y {e_2}\\
      &
      \Mid \kw{fold}_A(e) \Mid \kw{unfold}_A(e)
    \end{aligned}
  \end{align*}
In the syntax, $\alpha$ ranges over type variables, $x$ over expression variables, 
and~$n$ over natural numbers as constants. 
We identify terms up to renaming of bound variables. The term 
$\tlet {e_1} {<x, y>} {e_2}$ binds the variables~$x$ and~$y$ in~$e_2$ and
$\tcase {e} x {e_1} y {e_2}$ binds variable~$x$ in~$e_1$ and variable~$y$ in~$e_2$.
The term
$\kw{iszero}(e)$ is intended to have type $\unit+\unit$, with $\kw{inl}(<>)$
representing true.

We remark that the type~$\VN$ is used solely to encode values of the source
type of natural numbers~$\NN$. 
For applications to compilation, one may be
interested in restricting the natural numbers to, say, 64-bit integers.
Such a restriction can be made without affecting the results in this paper. 

Target expressions are typed with a standard type system, see Figure~\ref{fig:targettyping}.
A judgement $\SeqTm{\Gamma}{e}{A}$ therein expresses that~$e$ has type~$A$ in
context~$\Gamma$, where~$\Gamma$ is a finite mapping from variables to target
types.

For convenience, we allow ourselves ML-like data type notation
for working with recursive types. For example, for a type of 
lists we may write 
\[
  \beta\ \kw{list} = \kw{datatype}\ \kw{nil}\ \kw{of}\ \kw{unit} \mid
  \kw{cons}\ \kw{of}\ \beta \times (\beta\ \kw{list})
\]
instead of $\mu \alpha. \, \kw{unit} + \beta \times \alpha$, as 
$\kw{cons}(x,l)$ is
more readable than
$\kw{fold}_{\mu \alpha. \, \kw{unit} + \beta \times \alpha}(\kw{inr}(<x, l>))$.

For the operational semantics of target expressions we define
a standard call-by-value small-step reduction relation. We use 
the concepts of \emph{target values} and \emph{evaluation contexts}:
\begin{align*}
   \text{Values:}&&
   v, w &::= <> 
   \Mid n
   \Mid <v, w>
   \Mid \kw{inl}(v) \Mid \kw{inr}(v) 
   \Mid \kw{fold}_A(v)
   \\
   \text{Evaluation Contexts:}&&
   C &\mathrel{::=} [] \Mid C + e
   \Mid v + C
   \Mid \kw{iszero}(C)
   \\
 &&&\mathrel{\phantom{::=}} \phantom{[]} \Mid <C, e> \Mid <v, C>
   \Mid \tlet {C} {<x, y>} {e}
   \\
   &&&\mathrel{\phantom{::=}} \phantom{[]} 
   \Mid \kw{inl}(C) \Mid \kw{inr}(C)
   \Mid \tcase {C} x {e_1} y {e_2}
\end{align*}
The small-step reduction relation is then defined to be the
smallest relation $\red$ satisfying the following clauses:
\begin{align*}
  n_1 + n_2 &\red n_3 \text{ if~$n_3$ is the sum of~$n_1$ and~$n_2$}\\
  \kw{iszero}(0) &\red \kw{inl}(<>)
  \\
  \kw{iszero}(n) &\red \kw{inr}(<>) \text{ if~$n$ is non-zero}
  \\
  \tlet{<v_1,v_2>}{<x,y>}{e} &\red e[v_1/x, v_2/y]
  \\
  \tcase{\kw{inl}(v)} x {e_1} y {e_2} &\red e_1[v/x] 
  \\
  \tcase{\kw{inr}(v)} x {e_1} y {e_2} &\red e_2[v/x] 
  \\
  \kw{unfold}(\kw{fold}(v)) & \red v
  \\
  C[e_1] &\red C[e_2] \text{ if $e_1 \red e_2$}
\end{align*}
\begin{proposition}
  For each\/ $\SeqTm{\ }{e}{A}$ there exists a unique value~$v$
  satisfying $e \red^* v$.
\end{proposition}

\begin{figure}[t]
  \begin{prooftree}
    \AxiomC{$\I x A$ in $\Gamma$}
    \UnaryInfC{$ \SeqTm{\Gamma}{x}{A} $}
    \AxiomC{}
    \UnaryInfC{$ \SeqTm{\Gamma}{<>}{\unit} $}
    \alwaysNoLine
    \BinaryInfC{}
  \end{prooftree}
  \begin{prooftree}
    \AxiomC{}
    \UnaryInfC{$ \SeqTm{\Gamma}{n}{\VN} $}
    \AxiomC{$ \SeqTm{\Gamma}{e_1}{\VN} $}
    \AxiomC{$ \SeqTm{\Gamma}{e_2}{\VN} $}
    \BinaryInfC{$ \SeqTm{\Gamma}{e_1 + e_2}{\VN} $}
    \AxiomC{$ \SeqTm{\Gamma}{e}{\VN} $}
    \UnaryInfC{$ \SeqTm{\Gamma}{\kw{iszero}(e)}{\unit + \unit} $}
    \alwaysNoLine
    \TrinaryInfC{}
  \end{prooftree}
  \begin{prooftree}
    \AxiomC{$ \SeqTm{\Gamma}{e_1}{A} $}
    \AxiomC{$ \SeqTm{\Gamma}{e_2}{B} $}
    \BinaryInfC{$ \SeqTm{\Gamma}{<e_1, e_2>}{A\times B} $}
    \AxiomC{$ \SeqTm{\Gamma}{e_1}{A\times B} $}
    \AxiomC{$ \SeqTm{\Gamma,\, \I x A,\,\I y B}{e_2}{C} $}
    \BinaryInfC{$ \SeqTm{\Gamma}{\tlet{e_1}{<x,y>}{e_2}}{C} $}
    \alwaysNoLine
    \BinaryInfC{}
  \end{prooftree}
  \begin{prooftree}
    \AxiomC{$ \SeqTm{\Gamma}{e}{A} $}
    \UnaryInfC{$ \SeqTm{\Gamma}{\kw{inl}(e)}{A + B} $}
    \AxiomC{$ \SeqTm{\Gamma}{e}{B} $}
    \UnaryInfC{$ \SeqTm{\Gamma}{\kw{inr}(e)}{A + B} $}
    \alwaysNoLine
    \BinaryInfC{}
  \end{prooftree}
  \begin{prooftree}
    \AxiomC{$ \SeqTm{\Gamma}{e}{A + B} $}
    \AxiomC{$ \SeqTm{\Gamma,\, \I x A}{e_1}{C} $}
    \AxiomC{$ \SeqTm{\Gamma,\, \I y B}{e_2}{C} $}
    \TrinaryInfC{$ \SeqTm{\Gamma}{\tcase {e} x {e_1} y {e_2}}{C} $}
  \end{prooftree}
  \begin{prooftree}
    \AxiomC{$ \SeqTm{\Gamma}{e}{A[\mu \alpha.\,A/\alpha]} $}
    \UnaryInfC{$ \SeqTm{\Gamma}{\kw{fold}_{\mu \alpha.A}(e)}{\mu \alpha. A} $}
    \AxiomC{$ \SeqTm{\Gamma}{e}{\mu \alpha. A} $}
    \UnaryInfC{$ \SeqTm{\Gamma}{\kw{unfold}_{\mu \alpha.A}(e)}{A[\mu \alpha.\,A/\alpha]} $}
    \alwaysNoLine
    \BinaryInfC{}
  \end{prooftree}
  \caption{Typing of Target Expressions}
  \label{fig:targettyping}
\end{figure}

Having defined target expressions, we are now ready to define target
\emph{programs}. These consist of a set of first-order function
definitions.
Fix an infinite set~$\LL$ of function labels.

\begin{definition}  
  \label{def:def}
  A \emph{function definition} for label~$f\in \LL$ is given by an
  equation of one of the two forms
  \begin{align*}
    f(x) &= g(e),
    &
    f(x) &= \tcase e y {g(e_1)} z {h(e_2)},
\end{align*}
  wherein~$g,h\in\LL$ and $e$, $e_1$ and $e_2$ range over target expressions. 
\end{definition}
We allow ourselves to use syntactic sugar, writing
$f()$ for $f(<>)$
and $f(x,y) = t$ for $f(z) = t[\tlet z {<x,y>} {x}/x, \tlet z {<x,y>}
{y}/y]$, for example.

\begin{definition}
  A~\emph{target program} $P=(i,D,o)$ 
  consists of a set~$D$ of function definitions
  together with a list~$i\in\LL^*$ of entry labels
  and a list~$o\in\LL^*$ of exit labels. Both~$i$
  and~$o$ must be lists of 
  pairwise distinct labels.
  The set~$D$ of definitions must contain at most
  one definition for any label and must not contain
  any definition for the labels in~$o$.
\end{definition}  
The list~$i$ assigns an order to
the function labels that may be used as entry points for the program
and~$o$ identifies external labels as return points.

We use an informal graphical notation for target programs, depicting 
for example the program
$(\textit{const}\,\, \textit{pow}\,\, \textit{f}, D,
\textit{const\_ret}\,\, \textit{pow\_ret})$
as shown below.\\
\begin{center}
\includegraphics{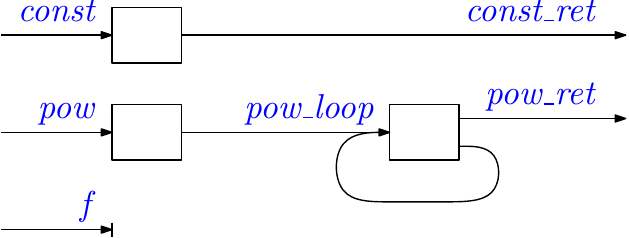}
\end{center}
The boxes correspond to the function definitions.
The arrows indicate for example that one may send a value~$v$ to
label~$\mathit{pow}$, which amounts to the function call $\mathit{pow}(v)$.
As a result a value will be sent to~$\mathit{pow\_loop}$. 

Target programs are well-typed if 
each function symbol~$f$ can be assigned an argument type~$A(f)$
such that each definition is well-typed:
A definition of the form $f(x) = g(e)$ is well-typed if
$\SeqTm{\I x {A(f)}}{e}{A(g)}$ is derivable; and 
a definition of the form
$f(x)= \tcase e y {g(e_1)} z {h(e_2)}$
is well-typed if 
$\SeqTm{\I x {A(f)}}{e}{C_1+C_2}$ and
$\SeqTm{\I y {C_1}}{e_1}{A(g)}$ and
$\SeqTm{\I z {C_2}}{e_2}{A(h)}$
are derivable for some~$C_1$ and~$C_2$.
If~$P$ is the program $(f_1\dots f_n, D, g_1\dots g_m)$, 
then we write $P\colon (A_1 \dots A_n) \to (B_1\dots B_m)$
if the argument types of $f_1,\dots, f_n,g_1,\dots, g_m$ are 
$A_1,\dots, A_n,B_1,\dots,B_m$
respectively.

We define a simple evaluation semantics for target programs.
A \emph{function call} is an expression of the form $f(v)$, 
where~$f$ is a function label and~$v$ is a value.
A relation $\red_P$ formalises the function calls
as they happen during the execution of a program~$P$.
It is the smallest relation satisfying the following 
conditions: if~$P$ contains
a definition $f(x)=g(e)$ 
then $f(v) \red_P g(w)$ for all values $v$ and~$w$ with $e[v/x] \red^* w$;
and if~$P$
contains a definition $f(x)=\tcase e y {g(e_1)} z {h(e_2)}$
then $f(v) \red_P g(w)$ 
for all values~$v$ and~$w$ with $\exists u.\, e[v/x] \red^* \kw{inl}(u) \land e_1[u/y]\red^* w$,
and $f(v) \red_P h(w)$
for all values $v$ and~$w$ with $\exists u.\, e[v/x] \red^* \kw{inr}(u) \land e_2[u/z]\red^* w$.

A \emph{call-trace} of program~$P$ is a sequence
$f_1(v_1)f_2(v_2)\dots f_n(v_n)$, such that 
$f_i(v_i) \red_P f_{i+1}(v_{i+1})$ holds for all $i\in \{1,\dots, n-1\}$.

\begin{definition}[Program Equality]
\label{def:eq}
Two programs 
$P, Q\colon (A_1\dots A_n) \to (B_1\dots B_m)$ 
are \emph{equal} if, for any input, they give the same output, that is, 
suppose the entry labels of~$P$ and~$Q$ are
$f_1,\dots,f_n$ and 
$g_1,\dots,g_n$ respectively and
the exit labels are 
$h_1,\dots,h_m$ and 
$k_1,\dots,k_m$ respectively, then,
for any $v$, $w$, $i$ and $j$, 
$P$ has a call-trace of the form $f_i(v)\dots h_j(w)$
if and only if $Q$ has a call-trace 
of the form $g_i(v)\dots k_j(w)$.
\end{definition}
Programs are thus equal, if the same input value on the same input
port leads to the same output value (if any) on the same output port in both
programs.

The following notation is used in Section~\ref{sect:base}.
For any list of target types $X=B_1\dots B_n$ and any target type~$A$,
we write $A\cdot X$ for the list $(A\times B_1)\dots (A\times B_n)$.
Given a program $P\colon X \to Y$, we write 
$A\cdot P \colon A\cdot X \to A\cdot Y$ for the program 
that passes on the value of type~$A$ unchanged and otherwise behaves like~$P$.
It may be defined by replacing each definition of the form $f(x)=g(e)$ in~$P$ 
with $f(u,x)=g(u,e)$ for a fresh variable~$u$, 
and each definition of the form $f(x) = \tcase e y {g(e_1)} z {h(e_2)}$
with $f(u,x) = \tcase e y {g(u,e_1)} z {h(u,e_2)}$, again for fresh~$u$.

We observe that target programs can be organised into a category~$\TT$ 
that has enough structure so that we can apply
the Int construction~\cite{Joyal96,Hasegawa09} (with respect to coproducts) to it and
obtain a category $\Int\TT$ that models interactive computation.

Target programs can be organised into a category\/~$\TT$. Its 
objects are finite lists of target types. A morphism 
from~$X$ to~$Y$ is given by a program $P \colon X\to Y$.
Two programs $P \colon X\to Y$ and $Q \colon X\to Y$ represent the same
morphism if and only if they 
are equal in the sense of Definition~\ref{def:eq}.
Thus, the morphisms from~$X$ to~$Y$ of~$\TT$ are the equivalence classes
of programs of type $X\to Y$ with respect to program equality.
\begin{lemma}
  $\TT$ is a category.
\end{lemma}
\begin{proof}[Proof outline]
  The identity on $X = A_1\dots A_n$ is the program 
  $(f_1\dots f_n,\emptyset, f_1\dots f_n)$.
  For the composition of $P\colon X \to Y$ and $Q\colon Y \to Z$,
  we first note that we can rename the labels~$P$ and~$Q$ such
  that we have $P=(i,D_P, m)$ and
  $Q=(m,D_Q,o)$.
  The composition $Q\circ P \colon X\to Z$ is then given simply by
  the program $(i, D_P\cup D_Q, o)$.
\end{proof}

\begin{lemma}
  The category\/~$\TT$ has finite coproducts, such that the initial object
  $0$ is given by the empty list and the object $X+Y$ is given by the
  concatenation of the lists~$X$ and~$Y$. Moreover, $\TT$ has 
  a uniform trace~\cite{Hasegawa09} with respect to these coproducts.
\end{lemma}
\begin{proof}[Proof outline]
  A simple proof can be given by observing that there is a faithful
  embedding from $\TT$ to the category of sets and partial functions. 
  The equations that are required to show then follow from the fact
  that the category of sets and partial functions has the desired
  structure.
\end{proof}
While we would like to emphasise the mathematical structure of 
target programs given by the Int construction, in the rest of the 
paper we shall spell it out concretely rather than referring to
categorical notions in order to make the paper easier to read.

\section{Source Language}
\label{sect:source}

Our source language is a variant of PCF,
a simply-typed $\lambda$-calculus with a basic type~$\NN$ 
of natural numbers and associated constants, as well as a 
fixed-point combinator for recursion. The intended evaluation
strategy is call-by-name.

The source language has the following types and terms.
\begin{align*}
  \text{Types: } && 
  X, Y &\ ::=\  1  \Mid X \to Y \Mid \NN  \\
  \text{Terms: } && s, t  &\ ::=\  
    * \Mid \tlami x X t \Mid s\ t 
     \Mid n \Mid s+t 
     \Mid \tif s {t_1} {t_2}
     \Mid \kw{fix}_X
\end{align*}
We write $\neg X$ as an abbreviation for the type $X\to \bot$.
Again, we identify terms up to renaming of bound variables. 

The typing judgement has the form $\SeqTm{\Gamma}{t}{X}$,
where $\Gamma$ is a finite list of variable declarations
$\I {x_1} {X_1},\dots, \I {x_n} {X_n}$.
We formulate the typing rules so that 
it is easy to consider fragments of the source language
of varying expressiveness.
The core rules are those of a linear $\lambda$-calculus 
and are given in Figure~\ref{fig:sourcecore}.
The rules for natural numbers appear in Figure~\ref{fig:sourcenat}.
We allow an addition operation $s+t$ instead of
the standard successor operation $\kw{succ}(s)$, as this gives a
simple example to explain the issues with the compilation of
multinary operation. 
Rules for contraction and for the fixed point combinator are 
given in Figures~\ref{fig:sourcecontr} and~\ref{fig:sourcefix}.

\begin{figure}
  \begin{prooftree}
    \AxiomC{}
    \LeftLabelSc{ax}
    \UnaryInfC{$ \SeqTm{\I x X}{x}{X} $}
    \AxiomC{}
    \LeftLabelSc{1i}
    \UnaryInfC{$ \SeqTm{}{*}{1} $}
    \alwaysNoLine
    \BinaryInfC{}
  \end{prooftree}
  \begin{prooftree}
    \AxiomC{$ \SeqTm{\Gamma}{t}{Y} $}
    \LeftLabelSc{weak}
    \UnaryInfC{$ \SeqTm{\Gamma,\, \I x X}{t}{Y} $}
    \AxiomC{$ \SeqTm{\Gamma,\, \I y Y,\, \I x X,\, \Delta}{t}{Z} $}
    \LeftLabelSc{exch}
    \UnaryInfC{$ \SeqTm{\Gamma,\, \I x X,\, \I y Y,\, \Delta}{t}{Z} $}
    \alwaysNoLine
    \BinaryInfC{}
  \end{prooftree}
  \begin{prooftree}
    \AxiomC{$ \SeqTm{\Gamma,\, \I x X}{t}{Y} $}
    \LeftLabelSc{$\to$i}
    \UnaryInfC{$ \SeqTm{\Gamma}{\tlami x X t}{X \to Y} $}
    \AxiomC{$ \SeqTm{\Gamma}{s}{X\to Y} $}
    \AxiomC{$ \SeqTm{\Delta}{t}{X} $}
    \LeftLabelSc{$\to$e}
    \BinaryInfC{$ \SeqTm{\Gamma,\, \Delta}{s\ t}{Y} $}
    \alwaysNoLine
    \BinaryInfC{}
  \end{prooftree}
  \caption{Source Language (\textsc{core}) – Linear Core}
  \label{fig:sourcecore}
  \vspace{1em}
  \begin{prooftree}
    \AxiomC{\phantom{X}}
    \LeftLabelSc{num}
    \UnaryInfC{$ \SeqTm{}{n}{\NN} $}
    \AxiomC{$ \SeqTm{\Gamma}{s}{\NN} $}
    \AxiomC{$ \!\SeqTm{\Delta}{t}{\NN} $}
    \LeftLabelSc{add}
    \BinaryInfC{$ \SeqTm{\Gamma,\, \Delta}{s + t}{\NN} $}
    \AxiomC{$ \SeqTm{\Gamma}{s}{\NN} $}
    \AxiomC{$ \!\SeqTm{\Delta_1}{t_1}{\NN} $}
    \AxiomC{$ \!\SeqTm{\Delta_2}{t_2}{\NN} $}
    \LeftLabelSc{if}
    \TrinaryInfC{$ \SeqTm{\Gamma,\, \Delta_1,\, \Delta_2}{\tif s {t_1} {t_2}}{\NN} $}
    \alwaysNoLine
    \TrinaryInfC{}
  \end{prooftree}
  \caption{Source Language (\textsc{lin}) – Natural Numbers}
  \label{fig:sourcenat}
  \vspace{2em}
  \centering
    \begin{prooftree}
      \AxiomC{$ \SeqTm{\Gamma,\,\I y X,\, \I z X}{t}{Y} $}
      \LeftLabelSc{contr}
      \UnaryInfC{$ \SeqTm{\Gamma,\,\I x X}{t[x/y, x/z]}{Y} $}
    \end{prooftree}
    \caption{Source Language (\textsc{stl}) – Contraction}
    \label{fig:sourcecontr}
  \vspace{2em}
  \begin{prooftree}
    \AxiomC{\phantom{$X$}}
    \LeftLabelSc{fix}
    \UnaryInfC{$ \SeqTm{}{\kw{fix}_X}{(X\to X)\to X} $}
  \end{prooftree}
  \caption{Source Language – Recursion}
  \label{fig:sourcefix}
\end{figure}

\section{CPS-Translation}
\label{sect:cps}

We use a variant of Hofmann and Streicher's 
\emph{call-by-name CPS-translation}~\cite{DBLP:conf/lics/HofmannS97},
which translates 
the source language extended with the following rules for product types as well as 
a type~$\bot$ without any rules.
\begin{prooftree}
  \AxiomC{$ \SeqTm{\Gamma}{s}{X} $}
  \AxiomC{$ \SeqTm{\Delta}{t}{Y} $}
  \LeftLabelSc{$\times$i}
  \BinaryInfC{$ \SeqTm{\Gamma,\, \Delta}{<s, t>}{X\times Y} $}
  \AxiomC{$ \SeqTm{\Gamma}{s}{X\times Y} $}
  \AxiomC{$ \SeqTm{\Delta,\, \I x X,\,\I y Y}{t}{Z} $}
  \LeftLabelSc{$\times$e}
  \BinaryInfC{$ \SeqTm{\Gamma,\, \Delta}{\tlet{s}{<x,y>}{t}}{Z} $}
  \alwaysNoLine
  \BinaryInfC{}
\end{prooftree}

For each source type~$X$, the type $\cps X$ of its continuations is defined
by:
\begin{align*}
  \cps{1} &= \neg 1 
  & \cps{\NN} &= \neg \NN
  & \cps{X \to Y} &= \neg \cps X \times \cps Y
\end{align*}  
A continuation for type $X\to Y$ is thus a pair
of a continuation of type $\cps Y$, using which the result can be
returned, and a function $\neg \cps X$ to access the argument.
A function can request its argument by applying this function to a
continuation of type~$\cps X$. The argument will then be provided to
this continuation.

For computation in continuation passing style, we often use the type
$\neg \cps X$, which we denote by $\overline X$.

The CPS-translation translates the source language into itself,
translating any typing derivation of
$\SeqTm{\I {x_1} {X_1},\dots,\I {x_n} {X_n}}{t}{Y}$
into a derivation of
$\SeqTm{\I {x_1} {\overline{X_1}},\dots,\I {x_n} {\overline{X_n}}}{\cps t}{\overline Y}$.
It is defined by induction on the given typing derivation.
Figure~\ref{fig:cps} shows how each typing rule on
the left is translated to a derived rule
on the right. 
\begin{figure}[t]
\begin{small}
\begin{center}
\begin{tabular}{ccc}
   \AxiomC{}
  \UnaryInfC{$ \SeqTm{\I x X}{x}{X} $}
  \DisplayProof
  &\ $\Longrightarrow$ &
   \AxiomC{}
  \UnaryInfC{$ \SeqTm{\I x {\overline X}}{\eta(x,\overline X)}{\overline X} $}
  \DisplayProof
\\[1em]
  \AxiomC{\phantom{X}}
  \UnaryInfC{$ \SeqTm{}{*}{1} $}
  \DisplayProof
  &\ $\Longrightarrow$ &
  \AxiomC{\phantom{X}}
  \UnaryInfC{$ \SeqTm{}
      {\lambda k.\, k\ *}
      {\overline 1} $}
  \DisplayProof
\\[1em]
    \AxiomC{$ \SeqTm{\Gamma}{t}{Y} $}
    \UnaryInfC{$ \SeqTm{\Gamma,\, \I x X}{t}{Y} $}
  \DisplayProof
  &\ $\Longrightarrow$ &
    \AxiomC{$ \SeqTm{\overline \Gamma}{\cps t}{\overline Y} $}
    \UnaryInfC{$ \SeqTm{\overline \Gamma,\, \I x {\overline X}}{\cps t}{\overline Y} $}
  \DisplayProof
\\[1em]
    \AxiomC{$ \SeqTm{\Gamma,\, \I y Y,\, \I x X,\, \Delta}{t}{Z} $}
    \UnaryInfC{$ \SeqTm{\Gamma,\, \I x X,\, \I y Y,\, \Delta}{t}{Z} $}
  \DisplayProof
  &\ $\Longrightarrow$ &
    \AxiomC{$ \SeqTm{\overline \Gamma,\, \I y {\overline Y},\, \I x {\overline
          X},\, \overline \Delta}{\cps t}{\overline Z} $}
    \UnaryInfC{$ \SeqTm{\overline \Gamma,\, \I x {\overline X},\, \I y
        {\overline Y},\, \overline \Delta}{\cps t}{\overline Z} $}
  \DisplayProof
\\[1em]
  \AxiomC{$ \SeqTm{\Gamma,\, \I x X}{t}{Y} $}
  \UnaryInfC{$ \SeqTm{\Gamma}{\tlami x X t}{X \to Y} $}
  \DisplayProof
  &\ $\Longrightarrow$ &
  \AxiomC{$ \SeqTm{\overline\Gamma,\, \I x {\overline X}}{\cps t}{\overline Y} $}
  \UnaryInfC{$ \SeqTm{\overline\Gamma}
     {\tlam {<x,k>} {{\cps t}\ {k}}}
   {\overline{X \to Y}} $}
  \DisplayProof
\\[1.5em]
  \AxiomC{$ \SeqTm{\Gamma}{s}{X\to Y} $}
  \AxiomC{$ \SeqTm{\Delta}{t}{X} $}
  \BinaryInfC{$ \SeqTm{\Gamma,\, \Delta}{s\ t}{Y} $}
  \DisplayProof
  &\ $\Longrightarrow$ &
  \AxiomC{$ \SeqTm{\overline\Gamma}{\cps s}{\overline{X \to Y}} $}
  \AxiomC{$ \SeqTm{\overline\Delta}{\cps t}{\overline X} $}
  \BinaryInfC{$ \SeqTm{\overline\Gamma,\, \overline\Delta}
      {\tlam k {{\cps s}\ {<\cps t, k>}}}
      {\overline Y} $}
  \DisplayProof
\\[1em]
  \AxiomC{\phantom{X}}
  \UnaryInfC{$ \SeqTm{}{n}{\NN} $}
  \DisplayProof
  &\ $\Longrightarrow$ &
  \AxiomC{\phantom{X}}
  \UnaryInfC{$ \SeqTm{}
      {\lambda k.\, k\ n}
      {\overline \NN} $}
  \DisplayProof
\\[1em]
  \AxiomC{$ \SeqTm{\Gamma}{s}{\NN} $}
  \AxiomC{$ \SeqTm{\Delta}{t}{\NN} $}
  \BinaryInfC{$ \SeqTm{\Gamma,\, \Delta}{s + t}{\NN} $}
  \DisplayProof
  &\ $\Longrightarrow$\ &
  \AxiomC{$ \SeqTm{\overline\Gamma}{\cps s}{\overline\NN} $}
  \AxiomC{$ \SeqTm{\overline\Delta}{\cps t}{\overline\NN} $}
  \BinaryInfC{$ \SeqTm{\overline\Gamma,\, \overline\Delta}
      {\lambda k.\, \cps{s}\ (\lambda x.\, \cps{t}\ (\lambda y.\, k\ (x+y)))}
      {\overline \NN} $}
  \DisplayProof
\\[1.5em]
  \AxiomC{$ \SeqTm{\Gamma}{s}{\NN} $}
  \AxiomC{\hspace{-1em}$ \SeqTm{\Delta_1}{t_1}{\NN} $}
  \AxiomC{\hspace{-1em}$\SeqTm{\Delta_2}{t_2}{\NN} $}
  \TrinaryInfC{$ \SeqTm{\Gamma,\, \Delta_1,\, \Delta_2}{\tif s {t_1} {t_2}}{\NN} $}
  \DisplayProof
  &\ $\Longrightarrow$\  &
  \AxiomC{$ \SeqTm{\overline\Gamma}{\cps s}{\overline\NN} $}
  \AxiomC{\hspace{-.2em}$ \SeqTm{\overline{\Delta_1}}{\cps{t_1}}{\overline\NN} $}
  \AxiomC{\hspace{-.2em}$ \SeqTm{\overline{\Delta_2}}{\cps{t_2}}{\overline\NN} $}
  \TrinaryInfC{$ 
        \begin{aligned}
          \overline\Gamma,\, \overline{\Delta_1},\,\overline{\Delta_2} \vdash
          \lambda k.\, \cps{s} \ (\lambda x.\, 
          \kw{if}\ x\ &\kw{then}\ {\cps{t_1}\ {(\lambda y.\, k \ y)}}
          \colon \overline \NN\\
          &\kw{else}\ {\cps{t_2} \ {(\lambda y.\, k \ y)}})
      \end{aligned}
      $}
  \DisplayProof
  \\[2.5em]
  \AxiomC{$ \SeqTm{\Gamma,\,\I y X,\, \I z X}{t}{Y} $}
  \UnaryInfC{$ \SeqTm{\Gamma,\,\I x X}{t[x/y, x/z]}{Y} $}
  \DisplayProof
  &\ $\Longrightarrow$ &
  \AxiomC{$ \SeqTm{\overline\Gamma,\,\I y {\overline X},\, \I z {\overline X}}{\underline t}{\overline Y} $}
  \UnaryInfC{$ \SeqTm{\overline\Gamma,\,\I x {\overline X}}{
  {\underline t[\eta(x,\overline{X})/y, \eta(x, \overline{X})/z]}}{\overline Y} $}
  \DisplayProof
  \\[2.5em]
  \AxiomC{\phantom{X}}
  \UnaryInfC{$ \SeqTm{}{\kw{fix}_X}{(X\to X) \to X} $}
  \DisplayProof
  &\ $\Longrightarrow$\ &
  \AxiomC{\phantom{X}}
  \UnaryInfC{$ \vdash
    \begin{aligned}[t]
       &\lambda <f,k>.\, \kw{fix}_{\overline X}\ (\lambda g.\lambda k_1.\,f\ <\lambda k_2.\,g\ k, \lambda x.\,k_1\ x>)\ k\\
       &\hspace{4.1cm}\colon {\overline{(X\to X) \to X}}
  \end{aligned}$}
  \DisplayProof
\end{tabular}

\end{center}
\end{small}
\caption{CPS-translation}
\label{fig:cps}
\end{figure}

This CPS-translation differs from the standard call-by-name 
CPS-translation of~\cite{DBLP:conf/lics/HofmannS97} in the use
of $\eta$-expansion in the rules for variables and contraction.
These expansions will allow us 
to use compositional reasoning in Sections~\ref{sect:linear}-\ref{sect:stl}. 
The term $\eta(t,X)$ is defined by induction on the type~$X$:
\begin{align*}
  \eta(t, X) &= t \text{ if $X$ is a base type ($1$, $\NN$ or $\bot$)}\\
  \eta(t, X\times Y) &= \tlet t {<x,y>} <\eta(x,X), \eta(y, Y)>\\
  \eta(t, X\to Y) &= \tlam x {\eta(t\ \eta(x, X), Y)} \text{ where~$x$ is fresh.}
\end{align*}
The last equation is more general than what we need in this paper.
We use only the special case $\eta(t, X\to \bot) = \tlam x {t\ \eta(x, X)}$,
as we apply $\eta$-expansion only to terms with types of the form $\overline Z$.
For example, we have
$\eta(x,\overline \NN) = \eta(x, \neg \neg \NN) 
= \lambda x_1.\,x\ (\lambda x_2.\,x_1\ x_2)$.

In the examples in the Introduction, we
have not applied this $\eta$-expansion for better readability.

\section{Defunctionalization}
\label{sect:defun}

In the translation from source to target, we first apply
CPS-translation and then use defunctionalization.
In this paper we use the flow-based defunctionalization 
procedure introduced by
Banerjee, Heintze and Riecke~\cite{banerjee}.
This procedure uses control flow information, so the CPS-translated
term is first annotated with control flow information and 
then defunctionalized using this information.

In this section we define a particularly simple special case
of the flow-based defunctionalization procedure.
It is too weak to handle the whole source language, but it
allows for a simple explanation of the relation to the Int construction.
We will extend the defunctionalization procedure
in Section~\ref{sect:stl} to cover the CPS-translation of the whole source language.

Control flow information is added to the terms in the form of
labelling annotations.
In the simple variant of the defunctionalization procedure that we
describe here, each function abstraction and application 
is annotated with a single label from $\LL$.
Thus, the terms $\tlami x X t$
and $s\ t$ are replaced by $\tlaml x X l t$ and $\tappl s l t$
respectively, where $l$ ranges over $\LL$. 
The function type $X\to Y$ is replaced by $X\xto{l} Y$,
again for any~$l \in \LL$.
We write $\neg_l X$ for $X \xto{l} \bot$.

We require that each abstraction be uniquely identified by its label, that is,
we allow only 
terms in which no two abstractions have the same label.
In the application $\tappl s l t$ the label~$l$ expresses that the
function~$s$ applied here is defined by an abstraction with
label~$l$.
The typing rules for abstraction and application are modified as
follows to enforce that terms are annotated
with correct control flow information. 
 \begin{prooftree}
   \AxiomC{$ \SeqTm{\Gamma,\, \I x X}{t}{Y} $}
   \centerAlignProof
   \UnaryInfC{$ \SeqTm{\Gamma}{\tlaml x X l t}{X \xto{l} Y} $}
   \AxiomC{$ \SeqTm{\Gamma}{s}{X\xto{l} Y} $}
   \AxiomC{$ \SeqTm{\Delta}{t}{X} $}
   \centerAlignProof
   \BinaryInfC{\raise.5em\hbox{$ \SeqTm{\Gamma,\, \Delta}{\tappl s l t}{Y} $}}
   \alwaysNoLine
   \BinaryInfC{}
\end{prooftree}

Allowing function types and applications to be annotated with a 
single label only is a real restriction.
For example, it is not possible to label and type terms such as
$\tlam x {<x\ (\tlam y 0), x\ (\tlam z 1)>}$.
The two abstractions $\tlam y 0$ and $\tlam z 1$  would each 
have to be given a unique label, say $l_1$ and $l_2$ respectively.
But then in the two uses of the variables~$x$, its types would have
to be ${X\xto{l_1} Y}$ and ${X\xto{l_2} Y}$ respectively.
With the above rules, this is not possible, as $l_1$ and $l_2$ are
different labels.
In general, one needs to 
allow types such as $X \xto{\{l_1,l_2\}} Y$ with 
more than a single label for more than one possible definition site, 
as in e.g.~\cite{banerjee}. We come back to this in Section~\ref{sect:stl}, but
up until then the variant with a single label suffices and simplifies the 
exposition.

In the rest of this section we explain how terms of
the labelled source language (with product types)
can be defunctionalized into programs in the target language. 
We defer the question of how to annotate the terms obtained by 
CPS-translation with labels to later sections (Lemmas~\ref{lem:cps} and~\ref{lem:cpsstl}).

The defunctionalization of a term~$t$ in the labelled source
language consists of a target expression $t^*$, which 
denotes the defunctionalized term itself, 
and a set of target equations~$D(t)$, which contains
the $\Capply$-definitions for defunctionalized function
application.
\begin{align*}
  x^* &= x\\
  n^* &= n\\
  (\tif s {t} {u})^* &= \tcase {\kw{iszero}(s^*)} {\_} {t^*} {\_} {u^*} \\
  <s,t>^* &= <s^*, t^*>\\
  (\tlet t {<x,y>} s)^* &= \tlet{t^*}{<x,y>}{s^*}\\
  (\tappl s l t)^* &= \textit{apply}_l(s^*, t^*) \\  
  (\tlaml x A l t)^* &= {<x_1,\dots, x_n>} \text{ where $\mathrm{FV}(\tlami x A t)=\{x_1,\dots, x_n\}$}
\end{align*}
In the last case for abstraction we assume some fixed global ordering on all
variables, so that the order of the tuple is well-defined.
\begin{align*}
  D(x) &= \emptyset\\
  D(n) &= \emptyset\\
  D(\tif s {t} {u}) &= D(s)\cup D(t)\cup D(u) \\
  D(<s,t>) &= D(s) \cup D(u) \\
  D(\tlet t {<x,y>} s) &= D(t) \cup D(s) \\
  D(\tappl s l t) &= D(s) \cup D(t) \\  
  D(\tlaml x A l t) &= D(t) \cup \{\textit{apply}_l(<x_1,\dots, x_n>, x)=t^*\} 
\end{align*}
In general, the set
$D(t)$ need not consist of function definitions in the strict
sense of Definition~\ref{def:def};
it may contain nested case distinctions, for example.
This technical issue could easily be solved in general at the expense
of making the technical development a little more complicated.
However, we shall use defunctionalization only for terms~$t$ for which $D(t)$ 
does in fact only consist of function definitions, so we
stick with the above simple definitions. 

Note that for closed terms of function type 
the target expression~$t^*$ is just~$<>$. Since 
all closed terms~$\cps t$ obtained by CPS-translation are of function
type, we therefore consider the definition set~$D(\cps t)$ as the
main result of defunctionalization.  

\begin{example}
\label{ex:defun}
With label annotations the example from the Introduction becomes the term
$\cps t$ given by
\[
  \lambda^{l_1} z.\, \tlet {z} {<x,k>} {
\tappl {(\lambda^{l_2}k'.\, \tappl {k'} {l_3} 1)} {l_2} {(\lambda^{l_3} u.\, \tappl x {l_5} {(\lambda^{l_4} n.\, \tappl k {l_6} {(u+n)}))}}}.
\]
Its type is ${\neg_{l_1}(\neg_{l_5} \neg_{l_4} \NN \times \neg_{l_6} \NN)}$. 
The set $D(\cps t)$
consists of the definitions
\begin{align*}
  \Capply_{l_1}(<>, <x,k>) &=\  \Capply_{l_2}(<>,<x,k>) ,
  &
  \Capply_{l_2}(<>, k') &=\  \Capply_{l_3}(k', 1) ,
  \\
  \Capply_{l_3}(<x,k>, u) &=\  \Capply_{l_5}(x, <k,u>) ,
  &
  \Capply_{l_4}(<k,u>, n) &=\ \Capply_{l_6}(k, u+n) .
\end{align*}
Compared to the definitions given in~(\ref{eq:defun}) in the Introduction, it appears that more
data is being passed around in these $\Capply$-equations.
However, consider once again the application of $\cps t$ to the
concrete arguments from the Introduction. Then one 
gets the additional
equations 
\begin{align*}
  \Capply_{l_5}(<>, k) &=\  \Capply_{l_4}(k, 42) ,
  &
  \Capply_{l_6}(<>, n) &=\ \texttt{print\_int}(n),
\end{align*}
and the fully applied term defunctionalizes to $\Capply_{l_1}(<>, <<>, <>>)$.
Thus, all the variables in the $\Capply$-equations only
ever store the value~$<>$ or tuples thereof, and these arguments 
may just as well be omitted.
\end{example}

An important point to note is that the defunctionalization procedure yields a set of
definition equations, but that it does not specify an interface of
entry and exit labels.
When one applies defunctionalization to 
a whole closed source programs of ground type, as is usually done in compilation,
choosing an interface is not important. One would
typically just choose a single entry label \texttt{main}
and a single exit label \texttt{exit}.
If one is interested in compositionality, however,
then open terms and 
terms of higher types must be also considered. 
Then one needs to fix an interface that explains how the 
free variables are accessed and how higher types are to be used.
In the above example term $\underline t$, a suitable choice of entry and
exit labels would be $l_1 l_4$ and $l_5 l_4$ respectively.
We shall explain how to define an interface from the
image of the CPS-translation in the next section.

Of course, the defunctionalization procedure described above is quite 
simple. In actual applications one would certainly want to apply optimisations, 
not least to remove unnecessary function arguments. An
example of such an optimisation is \emph{lightweight
defunctionalization} of Banerjee et al.~\cite{banerjee}. 
We shall argue that the Int construction captures one such optimisation
of the defunctionalization procedure. 

\section{The Core Linear Fragment}
\label{sect:linear}

To explain the basic idea of
how CPS-translation and defunctionalization
relate to a model of interactive computation
(namely $\Int\TT$),
we first consider the simplest non-trivial case.
We consider the core fragment of the source language, whose
syntax is 
\begin{align*}
 \text{Types: } && 
 X, Y &\ ::=\ 1  \Mid X \to Y  \\
 \text{Terms: } && s, t &\ ::=\  * 
  \Mid \tlami x X t \Mid s\ t 
\end{align*}
and whose rules are just those of Figure~\ref{fig:sourcecore}.
We call this source fragment \textsc{core}.

\subsection{Interactive Interpretation}

First we describe directly the interpretation of this fragment of 
\textsc{core} in $\Int\TT$.
A type~$X$ is interpreted by an interface $(X^-,X^+)$, which consists of
two finite lists $X^-$ and $X^+$ of target types.
Closed terms of type~$X$ will be interpreted as programs
of type $P\colon X^- \to X^+$.
The interfaces are defined by induction on the type:
\begin{align*}
  1^- &= \unit
  &
  (X \lollipop Y)^- &= Y^- X^+
  \\
  1^+ &= \unit
  &
  (X \lollipop Y)^+ &= Y^+ X^-
\end{align*}
Here, $X^- Y^-$ denotes the concatenation of the lists~$X^-$ and $Y^-$ 
(and likewise for the other cases). 

For a context $\Gamma = \I {x_1} {X_1},\dots, \I {x_n} {X_n}$,
we write $\Gamma^-$ and $\Gamma^+$ for the concatenations
$X^-_n\dots X^-_1$ and
$X^+_n\dots X^+_1$.

The interpretation of \textsc{core} is defined by induction on typing derivations.
A typing derivation of $\SeqTm{\Gamma}{t}{X}$ is interpreted by  
a morphism 
\[
  \sem{\SeqTm{\Gamma}{t}{X}}\colon X^-\Gamma^+ \to X^+\Gamma^-
\]
in $\TT$ (which amounts to a morphism from $(\Gamma^-,\Gamma^+)$ to 
$(X^-,X^+)$ in the category $\Int\TT$).
This interpretation is given in Figure~\ref{fig:intcore}.
\begin{figure}
\begin{center}
  \includegraphics{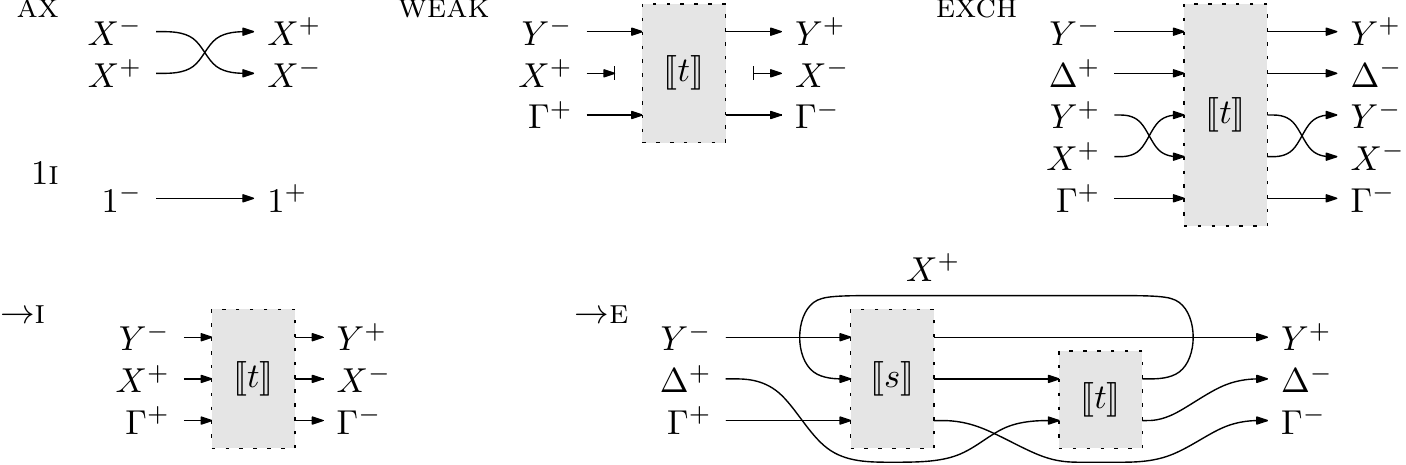}
\end{center}
\caption{Int-interpretation of $\textsc{core}$}
\label{fig:intcore}
\end{figure}
The boxes in this figure represent the inductive interpretation of the direct
sub-derivations of the individual rules.

It is a slight abuse of notation to write 
$\sem{\SeqTm{\Gamma}{t}{X}}$, even though the interpretation 
is defined not just from the sequent, but from its derivation.
We believe that it is possible to justify this notation by
proving that any two derivations
of the same sequent the same interpretation,
but in this paper we concentrate on the relation of
the interpretation to CPS-translation and defunctionalization
and always work with derivations.

\subsection{CPS-translation and Defunctionalization}

The aim is now to demonstrate that this interpretation in $\Int\TT$
is closely related to 
CPS-translation followed by defunctionalization.

To apply flow-based defunctionalization, we must find suitable labellings
of terms and types. We introduce special notation for labellings of types
of the form~$\overline X$.

\begin{definition}
For any type~$X$ and any $x^-,x^+\in \LL^*$ with
$\length{x^-} = \length{X^-}$
and
$\length{x^+} = \length{X^+}$,
we define a type $\overline X[x^-,x^+]$ in the 
labelled variant of \textsc{core} inductively as follows:
\begin{enumerate}
  \item Define $\overline{1}[q, a]$ to be $\neg_q \neg_a 1$.
  \item 
    If $\overline{X}[x^-,x^+]$ is defined and
    $\overline Y[y^-, y^+]$ is defined and of the form $\neg_{q} Y'$,
    then define $\overline{X\lollipop Y}[y^- x^+, y^+ x^-]$
    to be $\neg_{q}(\overline X[x^-, x^+] \times Y')$.
\end{enumerate}
\end{definition}
\noindent
For example, $\overline{1\lollipop 1}[qa',aq']$ denotes
$\neg_q(\neg_{q'}\neg_{a'}1 \times \neg_a 1)$.

Although $\overline X[x^-, x^+]$  is defined to be
abbreviation for a labelled type, 
one may alternatively think of it
as the type~$X$ together with a labelling of the ports of the interface 
$(X^-, X^+)$.

Readers familiar with game semantics may also want to compare the
syntax trees of the types $\overline{X}[x^-,x^+]$ with game semantic
arenas. 
The syntax tree induces a natural partial ordering on the labels appearing in
it:  $l_1 \lt l_2$ if there is a path from a node labelled $\xto{l_1}$ to one
labelled $\xto{l_2}$ in the syntax tree.
The Hasse diagrams of this ordering may be defined inductively as
follows:\\
\begin{center}
  \includegraphics{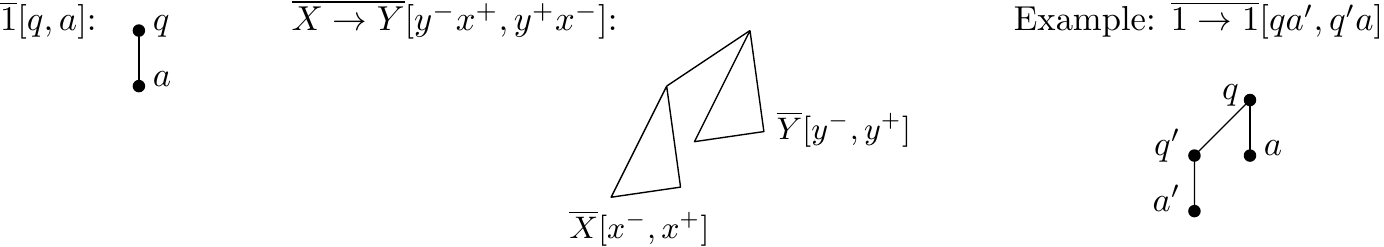}
\end{center}
These diagrams correspond to the game semantic arenas for the corresponding
types~\cite{HylandOng}.
More information about the relation of game arenas and continuations can be
found particularly in work of Levy~\cite{cbpv} and Melli\`es~\cite{mellies12}.

If $\Gamma$ is $\I {x_1} {X_1},\dots,\I {x_n} {X_n}$, then we write
short $\overline \Gamma[x_n^-\dots x_1^-, x_n^+\dots x_1^+]$ for
the context
$\I {x_1} {\overline{X_1}[x^-_1, x^+_1]}, 
\dots, 
\I {x_n} {\overline{X_n}[x^-_n,x^+_n]}$.
We say that a sequent 
$\SeqTm{\overline{\Gamma}[\gamma^-, \gamma^+]}{t}{\overline X[x^-,x^+]}$ 
is \emph{well-labelled} if the labels in $\gamma^-, \gamma^+, x^-,
x^+$ are pairwise distinct.

\begin{lemma}
  \label{lem:cps}
  If\/ $\SeqTm{\Gamma}{t}{X}$ is derivable in \textsc{core},
  then 
  the derivation of $\SeqTm{\overline \Gamma}{\cps t}{\overline X}$ obtained by CPS-translation 
  can be annotated
  with labels such that it derives the 
  well-labelled sequent\/
  $\SeqTm{\overline \Gamma [\gamma^-, \gamma^+]}{\cps t}{\overline X [x^-, x^+]}$
  for some $\gamma^-, \gamma^+, x^-, x^+ \in \LL^*$.
\end{lemma}
The proof is a straightforward induction on derivations. 
We note that the $\eta$-expansion 
in the CPS-translation of variables is essential for this lemma
to be true.  For example, with the $\eta$-expansion a well-labelled
$\SeqTm{\I x {\overline 1[q,a]}}{\cps x}{\overline 1[q',a']}$
is derivable;
without it this would only be possible if $q=q'$ and $a=a'$.
That $\eta$-expansions of variables can be labelled as needed follows
from the more general property established in the proof of Lemma~\ref{lem:eta} below.
The defunctionalization of $\cps x$ consists of definitions of
$\Capply_{q'}$ and $\Capply_{a}$, which just forward their arguments
to $\Capply_{q}$ and $\Capply_{a'}$ respectively.
We believe that it is simpler to consider the case with these
indirections first and study their removal (which is non-compositional,
due to renaming) in a possible second step.

We now define a function $\kw{CpsDefun}$ that combines CPS-translation
and defunctionalization. Given any \textsc{core}-derivation of a judgement 
$\SeqTm{\Gamma}{t}{X}$, let 
$\SeqTm{\overline \Gamma [\gamma^-, \gamma^+]}{\cps t}{\overline X [x^-, x^+]}$
be the 
judgement from the above lemma for a suitable choice of labels.
The function $\kw{CpsDefun}$ maps the source derivation of $\SeqTm{\Gamma}{t}{X}$ 
to the target program $(x^- \gamma^+, D(\cps t), x^+ \gamma^-)$,
where~$D(\cps t)$ is the set of equations obtained by
the defunctionalization of~$\cps t$.
It is not hard to see that the set~$D(\cps t)$ is indeed a 
target program whose definition does not depend on the choice of
labels. 

We define a single function $\kw{CpsDefun}$ rather than a composition of two
general functions $\kw{Cps}$ and $\kw{Defun}$, as in general there is no canonical
choice of entry and exit labels for defunctionalization. Thus,
the composition $\kw{Defun}\circ \kw{Cps}$ would only return a set of equations 
and not yet a target program. 
With a combined function, it suffices to choose entry and exit labels for 
terms that are in the image of the CPS-translation.

Define a further function $\kw{Erase}$ on target programs that
erases all function arguments.
\[
  \kw{Erase}(i, E, o) := 
   (i, \{f() = g() \mid f(x) = g(e)  \in E\}, o)
 \]
In fact, $\kw{Erase}$ also removes all equations
defined by case distinction, but these do not 
appear in $D({\cps t})$ for this source language.

The composition $\kw{Erase}\circ \kw{CpsDefun}$ of these two functions 
takes (a typing derivation of) a source program, applies the CPS-translation, defunctionalizes
and then `optimises' the result by erasing all function arguments.
The resulting program is in fact correct and it is what one obtains using the 
interpretation in $\Int\TT$:

\begin{proposition}
  \label{prop:simple}
  Suppose 
  $\SeqTm{\Gamma}{t}{X}$ is derivable in~\textsc{core}.
  Then the target program
  $\kw{Erase}(\kw{CpsDefun}(\SeqTm{\Gamma}{t}{X}))$ 
  has type $X^-\Gamma^+ \to X^+\Gamma^-$ and 
  defines the same morphism in~$\TT$ as the Int-interpretation 
  of\/ $\SeqTm{\Gamma}{t}{X}$.
\end{proposition}
\noindent
Since morphisms of type $X^-\Gamma^+ \to X^+\Gamma^-$ in~$\TT$ are 
defined to be equivalence classes of programs up to program equality 
(Definition~\ref{def:eq}),
the Int-interpretation $\sem{\SeqTm{\Gamma}{t}{X}}$ is an equivalence 
class of programs. The assertion of the proposition is therefore that the program
$\kw{Erase}(\kw{CpsDefun}(\SeqTm{\Gamma}{t}{X}))$ is an element of
the equivalence class $\sem{\SeqTm{\Gamma}{t}{X}}$.
\begin{proof}  
  The proof goes by induction on the derivation of 
  $\SeqTm{\Gamma}{t}{X}$.
  We continue by case distinction on the last rule in the derivation
  and show just the representative cases for variables and functions.
  \begin{itemize}
    \item Case \R{ax}.
      \begin{prooftree}
        \AxiomC{}
        \UnaryInfC{$ \SeqTm{\I x X}{x}{X} $}
      \end{prooftree}
  In this case $\kw{Erase}(\kw{CpsDefun}(\SeqTm{\Gamma,\,\I x X}{x}{X}))$ has the form 
  $(x_1^-  x_2^+,D, x_1^+ x_2^-)$ for
  \[
    D=\{\Capply_{x_1^-(i)}() = \Capply_{x_2^-(i)}(),
    \Capply_{x_2^+(i)}() = \Capply_{x_1^+(i)}() \mid i=1,\dots,n \},
  \]
  where we denote by~$w(i)$ the $i$-th element in the
  sequence~$w$ and where~$n$ is the common length of~$x_1^-$,
  $x_1^+$, $x_2^-$ and $x_2^+$.
  This is clearly in the equivalence class of the
  Int-interpretation.
\item Case \R{$\to$i}.
    \begin{prooftree}
      \AxiomC{$\vdots$}
      \noLine
      \UnaryInfC{$ \SeqTm{\Gamma,\, \I x X}{t}{Y} $}
      \UnaryInfC{$ \SeqTm{\Gamma}{\tlami x X t}{X\to Y} $}
    \end{prooftree}
  A CPS-translation of the derivation must have the following form
  in which 
  $Y[y^-,y^+]= \neg_{q_t} Y'$ and
  $y^- = q_t z$ 
  for
  $q_t \in \LL$
  and $z\in \LL^*$.
  \begin{prooftree}
    \AxiomC{\vdots}
    \noLine
    \UnaryInfC{$ \SeqTm{\overline\Gamma[\gamma^-,\gamma^+],\, \I x {\overline X[x^-,x^+]}}{\cps t}{\overline Y[y^-,y^+]} $}
    \AxiomC{}
    \UnaryInfC{$ \SeqTm{\I k {Y'}}{k}{Y'} $}
    \BinaryInfC{$ \SeqTm{\overline\Gamma[\gamma^-,\gamma^+],\, \I x {\overline X[x^-,x^+]}
      ,\, \I k {Y'}}
      {{\tappl{\cps t}{q_t}{k}}}
    {\bot} $}
    \UnaryInfC{$ \SeqTm{\overline\Gamma[\gamma^-,\gamma^+]}
      {\tlamlu {<x,k>} {\neg\cps X  \times \cps Y} q {\tappl{\cps t}{q_t}{k}}}
    {\overline {X\to Y}[q z x^+, y^+x^-]} $}
  \end{prooftree}
  By induction hypothesis,
  we know that 
  the
  program  $(y^-x^+\gamma^+, \kw{Erase}(D(\cps t)), y^+x^-\gamma^-)$
  is in the
  equivalence class of programs obtained by 
  Int-interpretation of the given derivation of 
  $\SeqTm{\Gamma,\,\I x X}{t}{Y}$. 

  We have to show that
  $(q zx^+\gamma^+,  
  \kw{Erase}(D(\tlamlu {<x,k>} {\neg\cps X  \times \cps Y} q {\tappl{\cps t}{q_t}{k}})),
    y^+x^-\gamma^-)$
  is in the
  equivalence class of programs obtained by 
  Int-interpretation of $\SeqTm{\Gamma}{\tlami x X t}{X\to Y}$. 
  But we have
  \[
   \kw{Erase}(D(\tlamlu {<x,k>} {\neg\cps X  \times \cps Y} q {\tappl{\cps t}{q_t}{k}}))
   =
   \{\Capply_q() = \Capply_{q_t}()\} \cup \kw{Erase}(D(\cps t))
  \]
  by definition. The definition of the Int-interpretation
  is such that the required assertion thus clearly holds.

\item Case \R{$\to$e}.
    \begin{prooftree}
      \AxiomC{\vdots}
      \noLine
      \UnaryInfC{$ \SeqTm{\Gamma}{s}{X\to Y} $}
      \AxiomC{\vdots}
      \noLine
      \UnaryInfC{$ \SeqTm{\Delta}{t}{X} $}
      \BinaryInfC{$ \SeqTm{\Gamma, \Delta}{s\ t}{Y} $}
    \end{prooftree}
  A CPS-translation of this derivation has the form 
  \begin{prooftree}
    \AxiomC{\vdots}
    \noLine
    \UnaryInfC{$ \SeqTm{\overline\Gamma[\gamma^-, \gamma^+]}{\cps s}{\overline{ X\lollipop Y}[y^-x^+, y^+x^-]} $}
    \AxiomC{\vdots}
    \noLine
    \UnaryInfC{$ \SeqTm{\overline\Delta[\delta^-, \delta^+]}{\cps t}{\overline X[x^-,x^+]} $}
    \AxiomC{ }
    \UnaryInfC{$ \SeqTm{\I k {Y'}}{k}{Y'} $}
    \BinaryInfC{$ \SeqTm{\overline\Delta[\delta^-, \delta^+],\, \I k {Y'}}{<\cps t, k>}{\overline X[x^-,x^+]} \times Y' $}
    \BinaryInfC{$ \SeqTm{\overline\Gamma[\gamma^-,\gamma^+],\, \overline\Delta[\delta^-,\delta^+]
      ,\, \I k {Y'}}
      {{\tappl{\cps s}{q_s}{<\cps t, k>}}}
    {\bot} $}
    \UnaryInfC{$ \SeqTm{\overline\Gamma[\gamma^-,\gamma^+],\, \overline\Delta[\delta^-,\delta^+]}
      {\tlamlu k {\cps Y} q {\tappl{\cps s}{q_s}{<\cps t, k>}}}
    {\overline Y[q z, y^+]} $}
  \end{prooftree}
  where $y^- = q_s z$ and $\overline Y[q z, y^+] = Y'$ and $q_s \in \LL$.

  Applying the induction hypothesis to the left and right
  sub-derivations of the given derivation shows that
  $(y^-x^+\gamma^+, \kw{Erase}(D(\cps s)), y^+x^-\gamma^-)$
  implements the Int-interpretation of 
  $\SeqTm{\Gamma}{s}{X\to Y}$
  and 
  $(x^-\delta^+, \kw{Erase}(D(\cps t)), x^+\delta^-)$
  implements the Int-interpretation of 
  $\SeqTm{\Delta}{t}{X}$.

  The program obtained by CPS-translation and
  defunctionalization is 
  \[
    (qz\delta^+\gamma^+, 
    \{\Capply_q() = \Capply_{q_s}()\} \cup D(\cps s) \cup D(\cps t),
    y^+\delta^-\gamma^-).
  \]
  By definition, $\kw{Erase}(D(\tlamlu k {\cps Y} q {\tappl{\cps s}{q_s}{<\cps t, k>}}))$
  has the form 
  $\{\Capply_q() = \Capply_{q_s}()\} \cup \kw{Erase}(D(\cps s)) \cup  \kw{Erase}(D(\cps t))$.
  This corresponds to the Int-interpretation of the sequent
  $\SeqTm{\Gamma,\Delta}{s\ t}{Y}$.\qedhere
  \end{itemize}
\end{proof}

\noindent While only for the very small source fragment~\textsc{core},
we have now seen how one can associate interfaces with higher-order types 
and show that the Int-interpretation implements these interfaces
in the same way as CPS-translation and defunctionalization. In 
Lemma~\ref{lem:cps} we have seen how $\eta$-expansion helps with
compositional reasoning.

\section{Base Types}
\label{sect:base}

We now work towards extending the result to a more expressive 
source language, starting with a fragment that extends \textsc{core} 
with non-trivial base types.
Define~\textsc{lin} to be the source fragment with
the syntax shown below and the 
typing rules from Figures~\ref{fig:sourcecore}
and~\ref{fig:sourcenat}.
\begin{align*}
  \text{Types: } && 
  X, Y &\ ::=\  1  \Mid X \to Y \Mid \NN  \\
  \text{Terms: } && s, t  &\ ::=\  * 
  \Mid \tlami x X t \Mid s\ t 
  \Mid n \Mid s+t 
  \Mid \tif s {t_1} {t_2}
\end{align*}
That is, we add the type of natural numbers~$\NN$ with
constant numbers, addition and case distinction,
but still consider only a linear source language.

The example in the Introduction shows that for \textsc{lin}
it is not possible to remove all arguments 
from the $\Capply$-functions, as we have done for \textsc{core}. 
At least certain natural numbers must be
passed as arguments.

\subsection{Interactive Interpretation}
Let us first consider the interpretation of \textsc{lin} in $\Int\TT$.
To this end we extend the definition of the interface $(X^-, X^+)$ as follows:
\begin{align*}
  1^- &= \unit
  &
  \NN^- &= \unit
  &
  (X \lollipop Y)^- &= Y^- X^+
  \\
  1^+ &= \unit
  &
  \NN^+ &= \VN
  &
  (X \lollipop Y)^+ &= Y^+ X^-
\end{align*}
The single value of type $\NN^-$ encodes the request 
to compute a particular number. 
The values of type $\NN^+$ are the possible answers.

It is not completely straightforward to extend the 
Int-interpretation described in the previous section.
Consider for example 
the case of an addition $s+t$ of two closed terms $\SeqTm{}{s}{\NN}$ and $\SeqTm{}{t}{\NN}$.
Suppose we already have programs $(q_s, D_s, a_s)$ and
$(q_t, D_t, a_t)$ for~$s$ and~$t$.
It is not possible to construct a program for~$s+t$ from these 
programs without modifying at least one of them. The problem is that after
evaluating the first summand, we have no way of storing the result while we invoke the
second program to compute the second summand.
A natural way of constructing a program for $s+t$ would be to take
the program $(q,D,a)$ with equations
$\Capply_q() = \Capply_{q_s}()$,
$\Capply_{a_s}(x) = \Capply_{q_t}(x,<>)$,
$\Capply_{a_t}(x,y) = \Capply_a(x+y)$,
the equations from~$D_s$, and the equations from~$\VN\cdot D_t$
(recall the notation $\VN\cdot -$ from Section~\ref{sect:target}).
Here we use~$\VN\cdot D_t$ instead of~$D_t$ 
in order to keep the value~$x$ of the first summand available until the 
second summand is computed, so that we can compute the sum.

One solution to this issue was proposed by Dal Lago and the author
in the form of IntML~\cite{intml}. We consider here a simple special case of this system. 
The basic idea is to annotate the domain of each function type $X\to Y$ 
with a \emph{subexponential}~$A$, which is a target type,
so that function types have the form $\lolli A X Y$.

We define \linexp, a variant of \textsc{lin} with subexponential
annotations. It has the same terms as \textsc{lin}, but the
grammar of types is modified as follows.
\[
  X, Y \ ::=\  1  \Mid \lolli A X Y \Mid \NN  
\]
In this grammar, $A$ ranges over target types.

The subexponential annotations may be explained such that a term~$s$ 
of type $\lolli A X Y$ 
is a function that uses its argument within an environment that
contains an additional value of type~$A$.
The function~$s$ may be applied to any argument~$t$ of type~$X$.
In the interactive interpretation, the application $s\ t$ is such that  
whenever~$s$ sends a query to~$t$, it needs to preserve a value of
type~$A$. It does so by sending the value along with the
query, expecting it to be returned unmodified along with a reply.
For example, addition naturally gets the type  
$\lolli \unit \NN {\lolli \VN \NN \NN}$, as it needs to remember the
already queried value  of the first argument (having type $\VN$) when it queries the
second argument.

It is interesting to note that 
Appel and Shao~\cite[\S 3.2]{DBLP:journals/toplas/ShaoA00} use a 
similar approach of preserving values by passing them as arguments 
for the optimisation of programs in CPS style.

Conceptually, subexponentials may be understood as a generalisation of the 
exponentials of Linear Logic. 
The special case $\lolli \omega X Y$, where the subexponential is the type
$\omega = \mu \alpha.\,\unit + \alpha$ of unbounded natural numbers, 
may be understood as ${!}X\to Y$.
This view corresponds to the construction of the exponential ${!}X$ in 
Game Semantics~\cite{AbramskyJM00} or in 
Geometry of Interaction situations~\cite{DBLP:journals/mscs/AbramskyHS02}.
We make the generalisation to subexponentials because it allows us to make only the assumptions
that are really needed, e.g.~with respect to assuming recursive types in 
the target language. It also allows us to avoid unnecessary encoding operations.
In the above outline of the translation of $s+t$, we could have used 
$\omega \cdot D_t$ instead of $\VN \cdot D_t$, but then
in the definition of $\Capply_{a_s}$ we would need to encode $x$ of
type~$\VN$ into a value of type $\omega$ and in the definition of
$\Capply_{a_t}$ we would need to decode again. 

The typing rules of \linexp are annotated version of the rules of
\textsc{lin}, formulated to keep track of subexponential annotations.
In \linexp contexts are finite lists of variable declarations of the form
$\J x A X$. The typing rules 
with subexponential annotations are are shown in
Figure~\ref{fig:linsubexp}. 
\begin{figure}
\begin{center}
  \begin{tabular}{cc}
    \AxiomC{$ $}
    \LeftLabelSc{ax}
    \UnaryInfC{$ \SeqU{\J x \unit X}{x}{X} $}
    \bottomAlignProof
    \DisplayProof
    &\quad
    \AxiomC{$ $}
    \LeftLabelSc{num}
    \UnaryInfC{$ \SeqU{}{n}{\NN} $}
    \bottomAlignProof
    \DisplayProof
  \end{tabular}
\end{center}
  \begin{prooftree}
    \AxiomC{$ \SeqTm{\Gamma}{t}{Y} $}
    \LeftLabelSc{weak}
    \UnaryInfC{$ \SeqTm{\Gamma,\, \J x A X}{t}{Y} $}
    \AxiomC{$ \SeqTm{\Gamma,\, \J y B Y,\, \J x A X,\, \Delta}{t}{Z} $}
    \LeftLabelSc{exch}
    \UnaryInfC{$ \SeqTm{\Gamma,\, \J x A X,\, \J y B Y,\, \Delta}{t}{Z} $}
    \alwaysNoLine
    \BinaryInfC{}
  \end{prooftree}
\begin{center}
  \begin{tabular}{cc}
    \AxiomC{$ \SeqU{\Gamma,\,\J x A X }{t}{Y} $}
    \LeftLabelSc{$\to$i}
    \UnaryInfC{$ \SeqU{\Gamma}{\lambda x: X.\, t}{\lolli A X Y} $}
    \bottomAlignProof
    \DisplayProof
    &\quad
    \AxiomC{$ 
    \SeqU{\Gamma}{s}{\lolli A X Y}
    $}
    \AxiomC{$ 
    \SeqU{\Delta}{t}{X}
    $}
    \LeftLabelSc{$\to$e}
    \BinaryInfC{$ \SeqU{\Gamma,\, A \cdot \Delta}{s\ t}{Y} $}
    \bottomAlignProof
    \DisplayProof
  \end{tabular}
\end{center}
\begin{center}
  \begin{tabular}{cc}
    \AxiomC{$ \SeqU{\Gamma }{s}{\NN} $}
    \AxiomC{$ \SeqU{\Delta}{t}{\NN} $}
    \LeftLabelSc{add}
    \BinaryInfC{$ \SeqU{\Gamma,\, \VN\cdot \Delta}{s+t}{\NN} $}
    \bottomAlignProof
    \DisplayProof
&\quad
    \AxiomC{$ \SeqW{\Gamma}{s}{\NN} $}
    \AxiomC{$ \SeqU{\Delta_1}{t_1}{\NN} $}
    \AxiomC{$ \SeqU{\Delta_2}{t_2}{\NN} $}
    \LeftLabelSc{if}
    \TrinaryInfC{$ \SeqU{\Gamma,\,\Delta_1,\,\Delta_2}
        {\tif s {t_1} {t_2}}{\NN} $}
    \bottomAlignProof
    \DisplayProof
  \end{tabular}
\end{center}
\caption{\linexp – \textsc{lin} with subexponential annotations}
\label{fig:linsubexp}
\end{figure}
In these rules, we write $A\cdot \Gamma$ for the context obtained by
replacing each declaration $\J x B X$ with $\J x {(A\times B)} X$.

With subexponential annotations, it is straightforward to define
the Int-interpretation.
Extend the definition of $(-)^-$ and $(-)^+$ to \linexp
by 
\begin{align*}
  (\lolli A X Y)^- &= Y^- (A\times X^+) 
  &
  \Gamma^- &= A_n\times X^-_n\dots A_1\times X^-_1
  \\
  (\lolli A X Y)^+ &= Y^+ (A\times X^-)
  &
  \Gamma^+ &= A_n\times X^+_n\dots A_1\times X^+_1,
\end{align*}
where $\Gamma$ is
$\J {x_1} {A_1} {X_1},\dots,\J {x_n} {A_n} {X_n}$.

The interpretation of the rules is shown graphically in
Figure~\ref{fig:intlin}.
The interpretation of rule~\R{ax} remains essentially the same,
but now uses the isomorphism $\unit \times A \iso A$ to treat the
subexponential.
The cases for \R{$\to$i} and \R{$\to$e} must also be modified 
to take subexponentials into account. In the case for \R{$\to$e} 
the box labelled with~$A$ represents the
program obtained by applying the operation $A\cdot (-)$
to the content of the box.  
In this case we moreover make the isomorphisms 
$(A\cdot \Delta)^+ \iso A\cdot \Delta^+$ 
and
$(A\cdot \Delta)^- \iso A\cdot \Delta^-$ 
implicit.
In the cases for \R{add} and \R{if} we omit the contexts
$\Gamma$, $\Delta$, $\Delta_1$ and $\Delta_2$ for better readability.
They are handled as in the case for \R{$\to$e}.
We omit the rules for pairs, which are also modified like the ones for 
functions~\cite{intml}.
In the case for \R{if}, we write ``$0?$'' for the program given by
$\Capply_{a_s}(x) = \tcase{\kw{iszero}(x)}{y}{\Capply_{q_1}(y)}{z}{\Capply_{q_2}(z)}$.
\begin{figure}
\begin{center}
  \includegraphics{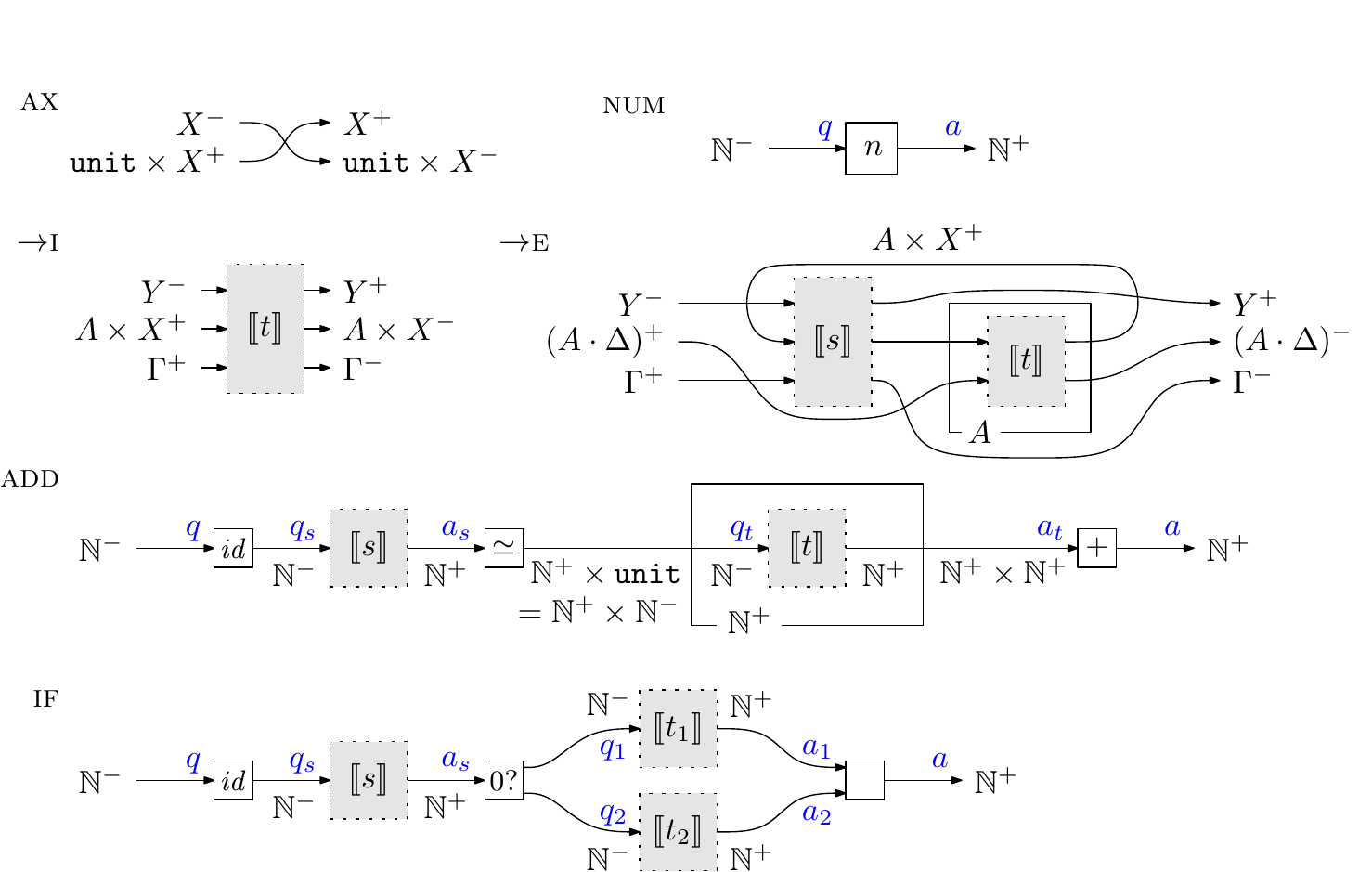}
\end{center}
\caption{Int-interpretation of \linexp}
\label{fig:intlin}
\end{figure}
A concrete definition of the Int-interpretation in terms of 
target equations can also be found in the proof of
Proposition~\ref{prop:skeleton} below.

\subsection{CPS-translation and Defunctionalization}

Let us now outline how this interpretation using the Int-construction relates to the
translation given by CPS-translation and defunctionalization, 
wherein the subexponential annotations are ignored.

A constant number~$n$ has the CPS-translation $\lambda^q k.\, \tappl k a n \colon
\overline \NN[q,a]$, where $\overline\NN[q,a]=\neg_q\neg_a \NN$. 
This defunctionalizes to $\Capply_q(<>, k)=\Capply_a(k,n)$. 
The Int-interpretation yields the definition
$\Capply_q() = \Capply_a(n)$, which differs 
only in that arguments have been removed.

For addition~$s+t$ a CPS-translation is 
$\lambda^q k.\, \tappl {\cps{s}} {q_s} {(\lambda^{a_s} x.\, \tappl{\cps{t}} {q_t} {(\lambda^{a_t} y.\, \tappl k a {(x+y)})})}$.
Defunctionalization leads to the following set of equations.
For the sake of illustration we assume that~$s$ and~$t$ are closed.
\[
D(\cps{s}) \cup D(\cps{t}) \cup
\{
  \begin{aligned}[t]
    &\Capply_q(<>, k) =\Capply_{q_s}(\cps s^*, <k>), \\
    &\Capply_{a_s}(<k>, x) =\Capply_{q_t}(\cps t^*, <k ,x>),\\
    &\Capply_{a_t}(<k,x>, y) =\Capply_{a}(k, x+y) \} .
\end{aligned}
\]
The program obtained in this way has the same shape as 
the program obtained by Int-interpretation. 
The program in Figure~\ref{fig:intlin} is annotated with labels to 
show the correspondence to the equations. 
The programs are not exactly equal. For example, $\Capply_q$ takes
a pair as an argument, 
while the program obtained by Int-interpretation expects a single
value of type~$\NN^-$. We study the relation of the two programs in the
rest of this section.

In a similar manner, the term 
$\tif s {t_1} {t_2}$ 
is CPS-translated to the labelled term
$
\lambda^q k.\, 
   \tappl {\cps{s}} {q_s} 
      {(\lambda^{a_s} x.\, 
         \tif x 
            {\tappl {\cps{t_1}} {q_1} {(\lambda^{a_1} y.\, \tappl k a y)}} 
            {\tappl {\cps{t_2}} {q_2} {(\lambda^{a_2} y.\, \tappl k a y)}})}$.
Defunctionalization gives us the equations 
\[
  D(\cps{s}) \cup D(\cps{t_1}) \cup D(\cps{t_2}) \cup 
\{
  \begin{aligned}[t]
    &\Capply_q(<>, k) =\Capply_{q_s}(\cps s^*, <k>), \\
    &\Capply_{a_s}(<k>, x) = 
      \tmcase {\kw{ifzero}(x)}
      {\_} {\Capply_{q_1}(\cps{t_1}^*, <k>)}
      {\_} {\Capply_{q_2}(\cps{t_2}^*, <k>)}
      \\
      &\Capply_{a_1}(<k>, y) =\Capply_{a}(k, y) ,\\
      &\Capply_{a_2}(<k>, y) =\Capply_{a}(k, y) \},
\end{aligned}
\]
and it can be observed that they correspond to 
the Int-interpretation given in Figure~\ref{fig:intlin}.

The observation that the programs obtained by Int-interpretation and
CPS-translation followed by defunctionalization have the same shape can be
made precise as follows.
\begin{definition}
  We say that two target programs have the same \emph{skeleton} whenever they
  have the same interface and the
  following holds:
  if one of the programs contains the definition $f(x)=g(e)$, then the other contains $f(x)=g(e')$ 
  for some~$e'$; and if one of the programs contains 
  $f(x)=\tcase e x {g(e_1)} y {h(e_2)}$, then 
  the other contains 
  $f(x)=\tcase {e'} x {g(e'_1)} y {h(e_2')}$
  for some~$e'$, $e'_1$ and $e'_2$. 
\end{definition}

We note that for \linexp Lemma~\ref{lem:cps} 
continues to hold and that $\kw{CpsDefun}$ can be defined 
exactly as for \textsc{core} above.

\begin{proposition}
  \label{prop:skeleton}
  For any derivation of~$\SeqTm{\Gamma}{t}{X}$ in \linexp
  there exists a program $\semc{\SeqTm{\Gamma}{t}{X}}$ that is a representative 
  of the Int-interpretation\/ $\sem{\SeqTm{\Gamma}{t}{X}}$ (which
  is a morphism in~$\TT$ and as such
  an equivalence class of programs up to program equality)
  and that has the same skeleton as $\kw{CpsDefun}(\SeqTm{\Gamma}{t}{X})$.
\end{proposition}
\proof
  Recall that $\kw{CpsDefun}$ first translates the derivation of
  $\SeqTm{\Gamma}{t}{X}$ to a labelled derivation of the
  CPS-translated term 
  $\SeqTm{\Gamma[\gamma^-, \gamma^+]}{\cps t}{X[x^-,x^+]}$,
  which is then mapped to the program
  $(x^-\gamma^+, D(\cps t), x^+\gamma^-)$.

  Here we show how to translate the derivation of
  $\SeqTm{\Gamma[\gamma^-, \gamma^+]}{\cps t}{X[x^-,x^+]}$
  to a set of equations $I(t)$, such that 
  the assertion of the proposition is satisfied when we choose
  $\semc{\SeqTm{\Gamma}{t}{X}} := (x^-\gamma^+, I(t), x^+\gamma^-)$.
  The definition of $I(t)$ is given by induction on the original 
  derivation by the following clauses:
  \begin{itemize}
    \item Rule \textsc{ax}. 
      \begin{prooftree}
        \AxiomC{\phantom{X}}
        \UnaryInfC{$ \SeqTm{\I x {\overline{\unit \cdot X}[q_1\dots q_n,a_1\dots a_n]}}{\eta(\cps x, \overline X)}{\overline X[q'_1\dots q'_n,a'_1\dots a'_n]} $}
        \bottomAlignProof
      \end{prooftree}
      We define 
      $I(x) := 
      \{\Capply_{q'_i}(x)=\Capply_{q_i}(<>,x) \mid i=1,\dots,n\}
      \cup
      \{\Capply_{a_i}(<>,x)=\Capply_{a'_i}(x) \mid i=1,\dots,n\} $

    \item
      Rule \textsc{num}.
      \begin{prooftree}
       \AxiomC{}
       \UnaryInfC{$ \SeqTm{}{\cps n}{\overline \NN[q,a]} $}
     \end{prooftree}
     Define $I(n) := \{\textit{apply}_q() = \textit{apply}_a(n)\}$.

   \item Rule \textsc{$\to$i}.
      \begin{prooftree}
       \AxiomC{$ \SeqTm{\overline \Gamma,\, \I x {\overline {A\cdot X}[x^-, x^+]}}{\cps t}{\overline {Y}[q_tz, y^+]} $}
       \UnaryInfC{$ \SeqTm{\overline \Gamma}{\cps{\tlam x t}}{\overline {(A\cdot X \lollipop Y)}[q z x^-, y^+ x^+]} $}
     \end{prooftree}
   Define $I(\tlam x t) := I(t) \cup \{ \textit{apply}_q() = \textit{apply}_{q_t}() \}$.

   \item Rule \textsc{$\to$e}.
     \begin{prooftree}
       \AxiomC{$ \SeqTm{\overline \Gamma}{\cps s}{\overline{(A\cdot X\lollipop Y)}[q_s z x^+, y^+ x^-]} $}
       \AxiomC{$ \SeqTm{\overline \Delta[\delta^-,\delta^+]}{\cps t}{\overline X[x^-, x^+]} $}
       \BinaryInfC{$ \SeqTm{\overline{\Gamma,\,A\cdot \Delta}}{\cps {s\ t}}{\overline Y[q z, y^+]} $}
     \end{prooftree}
     In the Int-interpretation we must account for the
     isomorphisms ${(A\cdot \Delta)^-} \iso A\cdot \Delta^-$ and
     $(A\cdot \Delta)^+ \iso A\cdot \Delta^+$.
     The set of equations $A\cdot I(t)$ gives rise to a program of type 
     $(A\cdot X^-) (A\cdot \Delta^+) \to (A\cdot X^+) (A\cdot \Delta^-)$.
     It is easy to define from it a program of type
     $(A\cdot X^-) (A\cdot \Delta)^+ \to (A\cdot X^+) (A\cdot \Delta)^-$:
     Each definition $\Capply_{q}(u, x) = e$  for $q\in \delta^-$ is 
     replaced by $\Capply_q(<u,v>, y) = e[<v,y>/x]$; and
     each call $\Capply_a(u, x)$ for $a\in \delta^+$ is replaced by
     $\Capply_a(<u, \kw{fst}(x)>, \kw{snd}(x))$, where $\kw{fst}$ and
     $\kw{snd}$ are the evident projections.
     Write $I(A\cdot t)$ for the program obtained in this way.

     With this notation, we can conclude this case by defining
     \[
        I(s\ t):= I(s) \cup I(A\cdot t) \cup\{ \textit{apply}_q() = \textit{apply}_{q_s}() \}.
     \]
   \item Rule \textsc{add}.
     \begin{prooftree}
       \AxiomC{$ \SeqTm{\overline \Gamma}{\cps s}{\overline \NN[q_s,a_s]} $}
       \AxiomC{$ \SeqTm{\overline \Delta}{\cps t}{\overline \NN[q_t,a_t]} $}
       \BinaryInfC{$ \SeqTm{\overline{\Gamma,\, \VN\cdot \Delta}}{\cps {s+t}}{\overline \NN[q,a]} $}
     \end{prooftree}
     We use the notation $I(\VN\cdot t)$, which is as in the case for
     \R{$\to$e} above, and define:
     \[
       I(s+t) := 
         I(s) \cup I(\VN\cdot t) \cup
       \{
       \begin{aligned}[t]
          &\textit{apply}_q() = \textit{apply}_{q_s}(), \\
          &\textit{apply}_{a_s}(x) = \textit{apply}_{q_t}(x,<>), \\
          &\textit{apply}_{a_t}(x,y) = \textit{apply}_{a}(x + y)\}
       \end{aligned}
     \]

   \item Rule \textsc{if}.
     \begin{prooftree}
       \AxiomC{$ \SeqTm{\overline \Gamma}{\cps s}{\overline \NN[q_s,a_s]} $}
       \AxiomC{$ \SeqTm{\overline{\Delta_1}}{\cps {t_1}}{\overline \NN[q_1,a_1]} $}
       \AxiomC{$ \SeqTm{\overline{\Delta_2}}{\cps {t_2}}{\overline \NN[q_2,a_2]} $}
       \TrinaryInfC{$ \SeqTm{\overline{\Gamma,\, \Delta_1,\, \Delta_2}}{\cps {\tif s {t_1} {t_2}}}{\overline \NN[q,a]} $}
     \end{prooftree}
     Let $I(\tif s {t_1} {t_2})$ be
     \[
       I_{\cps s} \cup I_{\cps {t_1}} \cup I_{\cps {t_2}} \cup
       \{\begin{aligned}[t]
           &\textit{apply}_q() = \textit{apply}_{q_s}(), \\
           &\textit{apply}_{a_s}(x) = \tmcase {\kw{iszero}(x)} y {\textit{apply}_{q_1}(y)} z {\Capply_{q_2}(z),} \\
           &\textit{apply}_{a_1}(x) = \textit{apply}_{a}(x),\\
           &\textit{apply}_{a_2}(x) = \textit{apply}_{a}(x)\}.
           \rlap{\hbox to 178 pt{\hfill\qEd}}
       \end{aligned}
    \]
  \end{itemize}

\noindent The proposition establishes a simple connection between the general shape
of the programs.

Let us now compare the values that are being passed around during program execution.
Consider closed terms of type~$\NN$. 
It follows by soundness of each of the two translations that 
the program obtained by defunctionalization and 
that for the Int-interpretation will return the same number as their end result.
If we consider the programs with the same skeleton constructed above,
then we can say more, however. We can show that during the computation
the two programs jump to the same labels in the same order. 
The argument values of these jumps are not exactly the
same, however. One may consider the values 
appearing in the program obtained by 
Int-interpretation as simplifications of the values
appearing at the same time in the traces of the program obtained by
defunctionalization. 
The following example illustrates the correspondence informally.

\begin{example}
Consider the source term $((\lambda x.\, 1 + x)\ 42)$ of type~$\NN$.
The result of CPS-translation and labelling is the term
\[
  \lambda^{l_0} k.\, \tappl{\cps t} {l_1} {<\lambda^{l_5} k''.\, \tappl {k''}
    {l_4} {42}, k>}
  \colon \neg_{l_0} \neg_{l_6} \NN \enspace,
\]
where $\cps t$ is spelled out in Example~\ref{ex:defun}. 
Defunctionalization gives us the following definitions.
\begin{align*}
  \Capply_{l_0}(<>, k) &=\  \Capply_{l_1}(<>,<<>,k>),
  &
  \Capply_{l_1}(<>, <x,k>) &=\  \Capply_{l_2}(<>,<x,k>) ,
  \\
  \Capply_{l_2}(<>, k') &=\  \Capply_{l_3}(k', 1) ,
  &
  \Capply_{l_3}(<x,k>, u) &=\  \Capply_{l_5}(x, <k,u>) ,
  \\
  \Capply_{l_4}(<k,u>, n) &=\ \Capply_{l_6}(k, u+n) ,
  &
  \Capply_{l_5}(<>, k'') &=\  \Capply_{l_4}(k'',42).
\end{align*}
The program of the same skeleton obtained by Int-interpretation is:
\begin{align*}
  \Capply_{l_0}() &=\  \Capply_{l_1}(),
  &
  \Capply_{l_1}() &=\  \Capply_{l_2}() ,
  \\
  \Capply_{l_2}() &=\  \Capply_{l_3}(1) ,
  &
  \Capply_{l_3}(m) &=\  \Capply_{l_5}(m) ,
  \\
  \Capply_{l_4}(m, n) &=\ \Capply_{l_6}(m+n),
  &
  \Capply_{l_5}(m) &=\  \Capply_{l_4}(m,42) .
\end{align*}
Both programs have entry label $l_0$ and exit label $l_6$.

Let us now compare how these programs compute their result.
A call trace of the first program, in which a closed continuation
represented by~$<>$ is given as argument, is:
\begin{align*}
&\Capply_{l_0}(<>,<>)\  
\Capply_{l_1}(<>,<<>,<>>) \ 
\Capply_{l_2}(<>,<<>,<>>)\ 
\Capply_{l_3}(<<>,<>>, 1)\ 
\Capply_{l_5}(<>, <<>, 1>)\ 
\\
&
\Capply_{l_4}(<<>,1>, 42)\ 
\Capply_{l_6}(<>, 43)
\intertext{%
The call trace of the second program is:}
&\Capply_{l_0}()\ 
\Capply_{l_1}() \
\Capply_{l_2}()\
\Capply_{l_3}(1)\
\Capply_{l_5}(1)\
\Capply_{l_4}(1, 42)\
\Capply_{l_6}(43)
\end{align*}
The point is that the traces are the same, up to simplification of values by
removing unneeded $<>$-values.
\end{example}

In the rest of this section we study the relation of the traces of the
programs obtained by the two translations. The example illustrates that the
traces of both programs jump to the same labels in the same order.
The main issue is to compare the argument values of each such jump.
We compare not the argument values themselves (keeping track of the technical
details appears to be non-trivial), but only what needs to be stored in order
to encode these values, i.e.~what a compiler needs to store in machine code.

For any target value~$v$, we define a multiset $\mathcal{V}(v)$ of the 
numbers it contains as follows: 
if $v=n$ then $\mathcal{V}(v) = \{n\}$,
if $v=<v_1,v_2>$ then
$\mathcal{V}(v) = \mathcal{V}(v_1) \cup \mathcal{V}(v_2)$,
and $\mathcal{V}(v) = \emptyset$ otherwise (values of
recursive types or sum types cannot appear).
The definition of~$\mathcal{V}(v)$ is motivated by considering how the
value~$v$ would eventually be encoded on a machine. A good compiler back-end
would need to store in memory only the values in~$\mathcal{V}(v)$, as the
rest of the information in~$v$ is given statically by the type.
We say that a value $v$ \emph{simplifies} a value $w$ if
$\mathcal{V}(v) \subseteq \mathcal{V}(w)$.
For example, the value $<2,<3,3>>$ simplifies $<1,<<2,<>>,<3,<2,3>>>>$,
but not $<2,3>$.
We say that a call trace $f_1(v_1)\dots f_n(v_n)$ \emph{simplifies} the
call trace $g(w_1)\dots g_n(w_n)$ if, for any $i\in\{1,\dots,n\}$,
$f_i=g_i$ and~$v_i$ simplifies~$w_i$.

With this terminology, we can 
express that the Int-interpretation of any term simplifies
its CPS-translation and defunctionalization in the sense that 
it differs only in that unused function arguments are removed and
function arguments are rearranged.

We shall analyse the behaviour of the program $\semc{\SeqTm{\Gamma}{t}{X}}$.
We use the notation $\semc{A\cdot(\SeqTm{\Gamma}{t}{X})}$
for the program obtained from $A\cdot \semc{\SeqTm{\Gamma}{t}{X}}$ 
by inserting the isomorphisms $(A\cdot \Gamma)^+ \to A \cdot
\Gamma^+$ and $A\cdot \Gamma^- \to (A\cdot \Gamma)^-$, as described
in the proof of Proposition~\ref{prop:skeleton} above
(case \R{$\to$e}).

\begin{theorem}
  \label{prop:trace}
  Let\/ $\SeqTm{}{t}{\NN}$, 
  let $(q, D_{\cps t}, a) := \kw{CpsDefun}(\SeqTm{\ }{t}{\NN})$
  and let\/ $\semc{\SeqTm{\ }{t}{\NN}}$ be the program from
  Proposition~\ref{prop:skeleton}.
  Then, any call-trace of\/ $\semc{\SeqTm{\ }{t}{\NN}}$ beginning with $\Capply_q()$ 
  simplifies the call-trace of\/~$\kw{CpsDefun}(\SeqTm{\ }{t}{X})$ of the same length
  that begins with $\Capply_q(<>, <>)$.
\end{theorem}
This theorem allows us to consider the Int-interpretation as a
simplification of the program obtained by defunctionalization.
This simplification seems quite similar to other optimisations of
defunctionalization, in particular lightweight
defunctionalization~\cite{banerjee}. However, we do not know 
any variant of defunctionalization in the
literature that gives exactly the same result.
One may consider the Int-interpretation
as a new approach to optimising the defunctionalization of
programs in continuation passing style.

To prove the theorem we use a few lemmas. The first two are
substitution lemmas.

\begin{lemma}
  \label{lem:substcpsdefun}
  If\/ $\SeqTm{\Gamma,\, \I x X}{s}{Y}$ and\/
  $\SeqTm{\ }{t}{X}$ are derivable in \textsc{lin},
  then so is\/
  $\SeqTm{\Gamma}{s[t/x]}{Y}$.
  Moreover,
  there exist a set of labels $E\subseteq \LL$
  and a bijective renaming $\rho\colon \LL \to \LL$,
  such that:
  If $\Capply_{l_1}(v_1)\dots \Capply_{l_n}(v_n)$
  is a trace of 
  $\kw{CpsDefun}(\SeqTm{\Gamma,\, \I x X}{s}{Y})\cup \kw{CpsDefun}(\SeqTm{\ }{t}{X})$,
  then 
  $c_1\dots c_n$
  is a trace of\/
  $\kw{CpsDefun}(\SeqTm{\Gamma}{s[t/x]}{Y})$, where
  \[
    c_i = 
    \begin{cases}
      \Capply_{\rho(l_i)}(v_i) & \text{if $l_i \notin E$,} \\
      \varepsilon & \text{otherwise.}
    \end{cases}
  \]
  Furthermore, all traces of\/ $\kw{CpsDefun}(\SeqTm{\Gamma}{s[t/x]}{Y})$
  arise in this way.
\end{lemma}
\begin{proof}[Proof outline]
This lemma is proved by induction on the derivation of 
$\SeqTm{\Gamma,\, \I x X}{s}{Y}$. 
The only interesting case is that where the
last rule is~\R{ax} and~$s$ is~$x$.
In this case the definitions in
$\kw{CpsDefun}(\SeqTm{\Gamma}{s[t/x]}{Y})$ and
$\kw{CpsDefun}(\SeqTm{\Gamma,\, \I x X}{s}{Y})\cup \kw{CpsDefun}(\SeqTm{\ }{t}{X})$
differ only in that the latter contains equations of the form
$\Capply_{l}(x) = \Capply_{l'}(x)$ that come from the $\eta$-expansion
of~$x$. The traces of the two programs thus differ only up to 
removal of these indirections. For the set~$E$ we choose the labels of
the calls that must be removed. A renaming~$\rho$ may be necessary to
deal with different choices of names in $\kw{CpsDefun}$.
\end{proof}

\begin{lemma}
  \label{lem:subst}
  If\/
  $\SeqTm{\Gamma,\, \J x A X}{s}{Y}$ and\/
  $\SeqTm{\ }{t}{X}$
  are derivable in \linexp, then so is\/
  $\SeqTm{\Gamma}{s[t/x]}{Y}$.
  Moreover,
  the set of labels $E\subseteq \LL$ and
  the bijective renaming $\rho\colon \LL \to \LL$
  from Lemma~\ref{lem:substcpsdefun} have 
  the following property:
  If $\Capply_{l_1}(v_1)\dots \Capply_{l_n}(v_n)$
  is a trace of\/ 
  $\semc{\SeqTm{\Gamma,\, \I x X}{s}{Y}}\cup \semc{A\cdot (\SeqTm{\ }{t}{X})}$,
  then\/ $\semc{\SeqTm{\Gamma}{s[t/x]}{Y}}$ has a trace 
  of the form $c_1\dots c_n$, where
  \[
    c_i = 
    \begin{cases}
      \Capply_{\rho(l_i)}(w_i) & \text{if $l_i \notin E$,} \\
      \varepsilon & \text{otherwise.}
    \end{cases}
  \]
  and\/ $\mathcal{V}(v_i) = \mathcal{V}(w_i)$.
  Furthermore, all traces of\/ $\semc{\SeqTm{\Gamma}{s[t/x]}{Y}}$
  arise in this way.
\end{lemma}
This lemma is again proved by induction on the derivation of
the first sequent. The statement is slightly weaker, as 
the traces of the two sets of equations may differ also
up to applications of the isomorphism $(\unit \times A) \iso A$,
as can be seen by considering the case where the last rule deriving 
$\SeqTm{\Gamma,\, \I x X}{s}{Y}$ is \R{ax} and $s$ is~$x$.
Thus, we only get $\mathcal{V}(v_i) = \mathcal{V}(w_i)$.

The next lemma says that any closed program of type~$\NN$ will
indeed eventually give an answer, as would already follow from soundness, and
moreover, the continuation that accepts the final answer is just passed 
along in the course of the computation; the computation itself
does not depend on the continuation.
\begin{lemma}
  \label{lem:trace}
  Let\/ $\SeqTm{}{t}{\NN}$ and let~$\cps t$ be labelled such that  
  $\SeqTm{}{\cps t}{\overline \NN[q,a]}$ is derivable. 
  Then the following are true.
  \begin{enumerate}
    \item
      Any call trace of $D(\cps t)$ beginning with 
      a call of the form $\Capply_q(<>, k)$, for some~$k$,
      can be extended to end with a call 
      $\Capply_a(k, v)$ for some value~$v$.
    \item
      If 
      $\Capply_q(<>, k_1) \dots \Capply_l(v_1, v_2)$
      and
      $\Capply_q(<>, k_2) \dots \Capply_{l'}(v'_1, v'_2)$
      are two call traces
      of $D(\cps t)$ of the same length,
      then $l = l'$ and there exist expressions $e_1$ and
      $e_2$ and variables $x_1$ and $x_2$, such that 
      $v_1 = e_1[k_1/x_1]$ and 
      $v_2 = e_2[k_2/x_2]$ holds.
  \end{enumerate}
\end{lemma}
\noindent
Note that the second point implies that if~$f(w)$ is a call in a call
trace beginning with $\Capply_q(<>, k)$, then~$k$ must simplify~$w$.
\begin{proof}
  In the proof we do not need subexponential annotations, so 
  we formulate it for \textsc{lin}.

  For each type~$X$ we define a set $R(X)$ of closed terms as follows: 
  $R(\NN)$ consists of all closed terms~$t$ that satisfy the assertion of
  the lemma;
  $R(X\to Y)$ consists of all closed terms~$s$ of type $X\to Y$ such that
  $t\in R(X)$ implies $s\ t \in R(Y)$.
  For any \textsc{lin}-context~$\Gamma$, we define $R(\Gamma)$ to be the
  set of all substitutions~$\sigma$ that map each variable declared in~$\Gamma$
  to a closed term, such that $\I x X \in \Gamma$ implies $\sigma(x)\in R(X)$. 

  The proof of the lemma then goes by showing by induction on the 
  derivation that each derivable $\SeqTm{\Gamma}{t}{X}$ 
  has the property $\forall \sigma \in R(\Gamma).\, t\sigma \in R(X)$. 
  The case for $\lambda$-abstraction follows using
  Lemma~\ref{lem:substcpsdefun}.
\end{proof}

\begin{proof}[Proof of Theorem~\ref{prop:trace}]
  The proof goes by induction on the size of the term~$t$.
  We continue by case distinction and consider representative cases.
  To simplify the notation, we just write $\semc{t}$ instead of
  $\semc{\SeqTm{\Gamma}{t}{X}}$.
  \begin{itemize}
    \item
      $t$ is $s_1+s_2$, i.e.~the derivation of $\SeqTm{}{t}{\NN}$ 
      ends with rule \R{add}.

      We observe that a labelling of the term $\cps{(s_1+s_2)}$ must have the
      following form 
      \[
      \SeqTm{}
      {\lambda^q k.\, \tappl {\cps{s_1}} {q_1} 
      {(\lambda^{a_1} x.\, \tappl {\cps{s_2}} {q_2} {(\lambda^{a_2} y.\, \tappl k a {(x+y)})})}}
      {\overline \NN[q,a]},
      \]
      where~$q$ and~$a$ are fresh and where $\cps{s_1}$ and $\cps{s_2}$
      are labelled such that 
      $\SeqTm{}{\cps{s_1}}{\overline{ \NN}[q_1, a_1]}$ and
      $\SeqTm{}{\cps{s_2}}{\overline{ \NN}[q_2, a_2]}$ are derivable.

      The program $D(\cps{s_1+s_2})$ consists of the set of equations
      \[
        D(\cps{s_1}) \cup D(\cps{s_2}) \cup \{
          \begin{aligned}[t]
            &\Capply_q(<>, k)=\Capply_{q_1}(<>, <k>),\\
            &\Capply_{a_1}(<k>, x)=\Capply_{q_2}(<>, <k,x>),\\
            &\Capply_{a_2}(<k,x>, y)=\Capply_{a}(k, x+y)\}.
          \end{aligned}
      \]

      On the other hand, the program $\semc{s_1+s_2}$ consists of
      the equations
      \[
        \semc{s_1} \cup \semc{\VN\cdot s_2} \cup \{
          \begin{aligned}[t]
           &\Capply_q()=\Capply_{q_1}(),\\
           &\Capply_{a_1}(x,<>)=\Capply_{q_2}(x),\\
           &\Capply_{a_2}(x,y)=\Capply_{a}(x+y)\}.
          \end{aligned}
      \]

      By the above Lemma~\ref{lem:trace}, we know that the call-trace of~$D(\cps t)$ 
      beginning with $\Capply_q(<>, <>)$ must have the form
      \[
      \Capply_q(<>, <>)\ \tau_1\ \tau_2\ \Capply_a(<>, x+y),
      \]
      where $\tau_1$ and $\tau_2$ must have the following forms:
      \begin{align*}
      \tau_1 &=\Capply_{q_1}(<>, <<>>)\dots \Capply_{a_1}(<>,x)\\
      \tau_2 &= \Capply_{q_2}(<>,<<>,x>) \dots \Capply_{a_2}(<<>,x>, y)
      \end{align*}
      Applying the induction hypothesis for~$s_1$ shows that
      the trace of~$\semc{s_1}$ starting with 
      $\Capply_{q_1}()$ simplifies 
      the trace of~$D(\cps t)$ starting with
      $\Capply_{q_1}(<>, <>)$.
      Using Lemma~\ref{lem:trace}, we get the desired property
      for~$\tau_1$.
      Similarly, the induction hypothesis for~$s_2$ shows that
      the trace of~$\semc{s_2}$ starting with 
      $\Capply_{q_2}()$ simplifies the trace of~$D(\cps t)$ starting with
      $\Capply_{q_2}(<>, <>)$.
      By Lemma~\ref{lem:trace}, the trace from 
      $\Capply_{q_2}(<>,<<>,x>)$ differs only in that it replaces
      $<>$ with $<<>,x>$ in each call and at
      least one position. But this shows
      that the trace of $\semc{\VN \cdot s_2}$
      starting with $\Capply_{q_2}(x)$ simplifies~$\tau_2$.
      Together this shows the desired property of the whole trace.
    \item
      $t$ cannot be a $\lambda$-abstraction, as its type is~$\NN$.
    \item
      $t$ is an application. In this case,~$t$ must have the 
      form $(\lambda x. s)\ t_1\ \dots\ t_n$, as it is a closed term.
      Notice that the term $s[t_1/x]\ t_2\ \dots\ t_n$ is shorter and
      $\SeqTm{}{s[t_1/x]\ t_2\ \dots\ t_n}{\NN}$ is still
      derivable.
      Hence, we can apply the induction hypothesis to it.

      It follows from Lemmas~\ref{lem:substcpsdefun}
      and~\ref{lem:subst} and the definition of the translations
      of $\lambda$-abstraction and application that 
      the desired result for $(\lambda x. s)\ t_1\ \dots\ t_n$
      follows from the result for $s[t_1/x]\ t_2\ \dots\ t_n$
      obtained by induction hypothesis.\qedhere
  \end{itemize}
\end{proof}

\section{Simple Types}
\label{sect:stl}

In this section, we strengthen the source language by adding
contraction, explain how the Int-interpretation can be extended 
and how it relates CPS-translation and defunctionalization.
With increasing 
expressiveness of the source language,
the syntactic details of defunctionalization become harder to
manage. For defunctionalization we now need a more expressive control
flow analysis, and the translation uses the recursive types in the 
target language.
We shall argue that a type system with subexponential
annotations, adapted from IntML, offers a simple
and conceptually clear way of 
managing such details.
We concentrate in this section only on the relationship between
program interfaces and skeletons.

We consider the source fragment \textsc{stl} of the simply-typed
$\lambda$-calculus with
the following syntax and the typing rules from
Figures~\ref{fig:sourcecore}--\ref{fig:sourcecontr}.
\begin{align*}
  \text{Types: } && 
  X, Y &\ ::=\  1  \Mid X \to Y \Mid \NN  \\
  \text{Terms: } && s, t  &\ ::=\   * 
  \Mid \tlami x X t \Mid s\ t 
  \Mid n \Mid s+t 
  \Mid \tif s {t_1} {t_2}
\end{align*}

\subsection{CPS-translation and Defunctionalization}

The CPS-translation defined in Section~\ref{sect:cps}
restricts to \textsc{stl}.
The defunctionalization procedure described
in Section~\ref{sect:defun}, however, is too simple to handle
contraction.
The control-flow annotations therein are not sufficient;
they need to be extended so that applications can be annotated with more than one label.

Banerjee et al.~\cite{banerjee} use a calculus with control flow
annotations, in which applications are annotated with sets of labels 
instead of just a single label. 
Thus, $\tappl s {\{l_1,\dots,l_n\}} t$ means that~$s$ is 
a term whose evaluation may have any of the functions with
label $l_1,\dots, l_n$ as a result.
Such an application is defunctionalized into 
a case distinction on the function that actually appears for~$s$ during
evaluation:
\[
(\tappl s {\{l_1,\dots,l_n\}} t)^* = \fcase {s^*} 
{l_1(\vec x)} {\Capply_{l_1}(l_1(\vec x), t^*) ;  \dots} 
{l_n(\vec y)} {\Capply_{l_n}(l_n(\vec y), t^*) }.
\]
Note that such a case distinction is possible only if labels are
actually passed as values.
To encode labels, one typically uses algebraic
data types whose constructors correspond to the function labels. 
To handle the full simply-types $\lambda$-calculus, one must allow for recursive
algebraic data types. An example is given in Example~\ref{ex:rec} on page \pageref{ex:rec}.

We define a variant of the labelled $\lambda$-calculus of 
Banerjee et al.~\cite{banerjee}, which is suitable
for the target language considered here (the target language in~\cite{banerjee} has
union types, while we use disjoint sums here).

\newcommand{\Coercl}[2]{\kw{coercl}_{#1}(#2)}
\newcommand{\Coercr}[2]{\kw{coercr}_{#1}(#2)}
\newcommand{\coercl}[2]{\kw{coercl}(#2)}
\newcommand{\coercr}[2]{\kw{coercr}(#2)}

Instead of sets of labels, we annotate applications with \emph{label terms}
formed by the following grammar, in which $l$ ranges over all the
labels from $\LL$.
\[
  L_1, L_2 \ ::=\  l \Mid L_1 + L_2
\]
Write $\LT$ for the set of all label terms.

In the labelled version of \textsc{stl} with product types, each abstraction
$\tlaml x X l t$
is still annotated with a unique label $l\in \LL$. 
Applications $\tappl s L t$, however, are now annotated with a label term.
Function types are also annotated with a
label term instead of just a single label.
Moreover, we extend the type system with 
explicit coercion terms $\Coercl{L}{t}$ and
$\Coercr{L}{t}$. 
The syntax of the labelled \textsc{stl} with products is therefore given
as follows.
\begin{align*}
  \text{Types: } && 
  X, Y &\ ::=\  1  \Mid X \xto L Y \Mid X\times Y\Mid \NN \Mid \bot \\
  \text{Terms: } && s, t  &\ ::=\  
  \begin{aligned}[t]
  * 
  &\Mid \tlaml x X l t \Mid \tappl s L t 
  \Mid <s, t> \Mid \tlet {s} {<x, y>} {t} 
  \\
  &\Mid n \Mid s+t 
  \Mid \tif s {t_1} {t_2} 
  \\
  &\Mid \Coercl{L}{t} \Mid \Coercr{L}{t}
\end{aligned}
\end{align*}
The typing rules for the new and modified terms are: 
 \begin{prooftree}
   \AxiomC{$ \SeqTm{\Gamma,\, \I x X}{t}{Y} $}
   \centerAlignProof
   \UnaryInfC{$ \SeqTm{\Gamma}{\tlaml x X l t}{X \xto{l} Y} $}
   \AxiomC{$ \SeqTm{\Gamma}{s}{X\xto{L} Y} $}
   \AxiomC{$ \SeqTm{\Delta}{t}{X} $}
   \centerAlignProof
   \BinaryInfC{\raise.5em\hbox{$ \SeqTm{\Gamma,\, \Delta}{\tappl s L t}{Y} $}}
   \alwaysNoLine
   \BinaryInfC{}
 \end{prooftree}
 \begin{prooftree}
   \AxiomC{$ \SeqTm{\Gamma}{t}{X \xto{L_1} Y} $}
   \centerAlignProof
   \UnaryInfC{$ \SeqTm{\Gamma}{\Coercl{L_1 + L_2}{t}}{X \xto{L_1 + L_2} Y} $}
   \AxiomC{$ \SeqTm{\Gamma}{t}{X \xto{L_2} Y} $}
   \centerAlignProof
   \UnaryInfC{$ \SeqTm{\Gamma}{\Coercr{L_1 + L_2}{t}}{X \xto{L_1 + L_2} Y} $}
   \alwaysNoLine
   \BinaryInfC{}
 \end{prooftree}
The type $X \xto{(l_1 + l_2) + l_3} Y$ thus is the type of functions with label
$l_1$, $l_2$ or $l_3$. Annotating this type with the term $(l_1 + l_2) + l_3$,
as opposed to the set $\{l_1, l_2, l_3\}$, is convenient for technical
reasons, as our target
language has disjoint sum types and not union types, which means 
that we cannot assume associativity.

We shall often omit the subscript~$L$ in the terms $\Coercl{L}{t}$ and
$\Coercr{L}{t}$, when it can be reconstructed from type information.

With these changes to the label annotations, we can extend the
defunctionalization procedure to cover the whole source language.
The new terms are defunctionalized as follows (the notation
$l(x_1,\dots, x_n)$ is explained below).
\begin{align*}
  (\tappl s L t)^* &= \Capply_L(s^*, t^*)\\
  (\tlaml x X l t)^* &= l(x_1,\dots, x_n) \text{ where $\mathrm{FV}(\tlaml x X l t)=\{x_1,\dots, x_n\}$} \\
  (\Coercl{L}{t})^* &= \kw{inl}(t^*) \\
  (\Coercr{L}{t})^* &= \kw{inr}(t^*) \\
\intertext{Definitions:}
  D(\tappl s L t) &= D(s) \cup D(t) \cup D(L)\\
  D(\tlaml x X l t) &= D(t) \cup \{ \Capply_l(l(x_1,\dots, x_n), x) = t^* \} \\
  D(\Coercl{L}{t}) &= D(t) \cup D(L)\\
  D(\Coercr{L}{t}) &= D(t) \cup D(L)
\end{align*}
where
\begin{align*}
  D(l) &= \emptyset \\
  D(L_1 + L_2) &= D(L_1) \cup D(L_2) \cup 
   \{\Capply_{L_1+L_2}(f,x) =
   \tmcase f {f_1} {\Capply_{L_1}(f_1, x)} {f_2} {\Capply_{L_2}(f_2, x)\}}
\end{align*}
The term $l(x_1,\dots, x_n)$ plays the role of a constructor in a 
functional language. For each label~$l \in \LL$ we assume a 
data type~$\tau_l$ with a single constructor called~$l$ with
arguments of appropriate type to make the above definition type correct. 
In ML-notation one would write
\[
  \tau_{l} = \kw{datatype}\ l\ \kw{of}\ A
\]
for a suitable type~$A$.
We extend the definition of $\tau_l$ to label terms by letting 
$\tau_{L_1+L_2} = \tau_{L_1} + \tau_{L_2}$.
The definition of these types is 
such that in a definition of $\Capply_L(f, x)$, 
the variable~$f$ will have type~$\tau_{L}$.

If one writes out all the data type definitions for a given term,
then one may obtain a set of mutually recursive data type definitions,
as the type~$\tau_l$ may appear in the argument type~$A$ of
its constructor (see Example~\ref{ex:rec}).
In cases where the definitions are not actually recursive, it would be
possible to remove the constructors and work just with tuples instead. 

\begin{example}
  Let us illustrate the modified defunctionalization
  by considering the CPS-translation of $\tlami x \NN {x + x}$.
  For this example, $\eta$-expansion is not important, so we 
  omit it for simplicity.
  The CPS-translated and simplified term
  may be annotated with label terms as follows: 
  \[
    \lambda^{l_1} <x,k>.\, 
      \tappl x {l_4} 
         {\coercl{l_2+l_3}{\lambda^{l_2} m.\, \tappl x {l_4} {\coercr{l_2+l_3}{\lambda^{l_3} n.\, \tappl k {l_5}{(m+n)}}}}}
  \]
  Its type is
  $(((\NN \xto{l_2 + l_3} \bot) \xto{l_4} \bot) \times (\NN \xto{l_5} \bot)) \xto{l_1} \bot$. 
  As a concrete argument one may think of 
  $<\lambda^{l_4} k.\, \tappl k {l_2 + l_3} 42,\, \lambda^{l_5} n.\, \texttt{print\_int}(n)>$.
  Note the use of the label term~$l_2 + l_3$. The two applications
  of~$x$ could not be typed using the simple labelled $\lambda$-calculus 
  from Section~\ref{sect:defun}.

  Defunctionalization turns the term into the following definitions.
  \begin{align*}
    \Capply_{l_1}(l_1(), <x,k>) &= \Capply_{l_4}(x, \kw{inl}({l_2}(x,k)))
    \\
    \Capply_{l_2}(l_2(x,k), m) &= \Capply_{l_4}(x, \kw{inr}(l_3(m,k)))
    \\
    \Capply_{l_3}(l_3(m,k),n) &= \Capply_{l_5}(k, m+n) 
    \\
    \Capply_{l_4}(l_4(), k) 
    &=  \Capply_{l_2 + l_3}(k, 42)\\
    \Capply_{l_2+l_3}(k, n)  
    &=
    \tmcase k {f_1} {\Capply_{l_2}(f_1, n)} 
    {f_2} {\Capply_{l_3}(f_2, n)} 
    \\
    \Capply_{l_5}(l_5(), n) &= \texttt{print\_int}(n)
  \end{align*}
The types of the constructors are: 
\begin{align*}
  \tau_{l_1} &= \kw{datatype}\ l_1\ \kw{of}\ \unit 
  &
  \tau_{l_2} &= \kw{datatype}\ l_2\ \kw{of}\ \tau_{l_4} \times \tau_{l_5} \\
  \tau_{l_3} &= \kw{datatype}\ l_2\ \kw{of}\ \VN \times \tau_{l_5} 
  &
  \tau_{l_4} &= \kw{datatype}\ l_4\ \kw{of}\ \unit \\
  \tau_{l_5} &= \kw{datatype}\ l_5\ \kw{of}\ \unit 
\end{align*}
In this example, these types are not actually recursive, so 
we could remove the constructors, replacing $l_2(x,k)$ just by
the tuple $<x,k>$, etc.
\end{example}

Next we show how any CPS-translated \textsc{stl}-term can be
suitably annotated with labels, 
so that defunctionalization can be applied. 

We carry over the notation $\overline X[x^-, x^+]$ from
Section~\ref{sect:linear}, but now allow $x^-$ and $x^+$ range over
$\LT^*$ instead of $\LL^*$.

In order to label the CPS-translated terms, we now have to deal with the new case
for contraction. 
Recall the CPS-translation of contraction from Figure~\ref{fig:cps}:
\begin{center}
\begin{tabular}{ccc}
  \AxiomC{$ \SeqTm{\Gamma,\,\I y X,\, \I z X}{t}{Y} $}
  \UnaryInfC{$ \SeqTm{\Gamma,\,\I x X}{t[x/y, x/z]}{Y} $}
  \DisplayProof
  &\ $\Longrightarrow$ &
  \AxiomC{$ \SeqTm{\overline\Gamma,\,\I y {\overline X},\, \I z {\overline X}}{\underline t}{\overline Y} $}
  \UnaryInfC{$ \SeqTm{\overline\Gamma,\,\I x {\overline X}}{
  {\underline t[\eta(x,\overline{X})/y, \eta(x, \overline{X})/z]}}{\overline Y} $}
  \DisplayProof
\end{tabular}
\end{center}
The use of~$\eta$ expansions
allows us to label the
CPS-translated terms in a compositional way,
much like in Lemma~\ref{lem:cps}.
For the sake of illustration, consider the case where~$\Gamma$ is empty and
where~$X$ is~$\NN$, so that $\overline X$ is $\neg \neg\NN$.
Suppose we have already labelled the premise of the CPS-translation, say as in
\[
\SeqTm{\I {y} {\neg_{q_1} \neg_{a_1} \NN},\, \I {z} {\neg_{q_2} \neg_{a_2}
    \NN}}{\cps t}{\overline Y[y^-,y^+]}.
\]
Note that the types of the variables~$y$ and~$z$ will in general be annotated with different
labels. This means that these variables have different types and we cannot use contraction to make them
into a single variable~$x$.
However, we can annotate the $\eta$-expansion of~$x$, i.e.~the term $\eta(x,
\neg\neg\NN)$, in the following two ways,
in which
$a'_1$ and $a'_2$ are fresh labels and~$q$ is any label term.
\begin{align*}
\SeqTm{\I x {\neg_q \neg_{a'_1+a'_2} \NN}}
{
  \lambda^{q_1} k.\, 
  \tappl x {q} 
  {\coercl {} {\lambda^{a_1'} y.\, \tappl k {a_1} y}}
}{\neg_{q_1} \neg_{a_1} \NN }
\\
\SeqTm{\I x {\neg_q \neg_{a'_1+a'_2} \NN}}
{
\lambda^{q_2} k.\, \tappl x {q} {\coercr {} {\lambda^{a_2'} y.\, \tappl k {a_2} y}}
}{\neg_{q_2} \neg_{a_2} \NN }
\end{align*}
If we substitute the first term for~$y$ and the second term for~$z$, then
the resulting term is a labelled version of
$\SeqTm{\I {x} {\neg_{q} \neg_{a'_1+a'_2} \NN}}{\cps t[\eta(x,\overline{X})/y,
  \eta(x, \overline{X})/z]}{\overline Y[y^-,y^+]}$,
i.e.~the conclusion of the CPS-translation of contraction.
This outlines how we can substitute $\eta$-expansions of~$x$, whereas
we could not just substitute~$x$ for both~$y$ and~$z$,

Of course, it remains to be shown that it is possible to find a labelling
of the whole term even when the type of~$x$ contains label terms, such as
${a'_1+a'_2}$ instead of just fresh labels.
Note that in the type of~$x$ we could not have put a single label in place of 
 ${a'_1+a'_2}$, as $a'_1$ and~$a'_2$ are the unique labels of  
two different abstractions.

The next two lemmas show that it is indeed always possible to label CPS-translated 
terms appropriately.
The following lemma first generalises the above labelling of $\eta$-expansion
from~$\NN$ to an arbitrary type~$X$. In the subsequent lemma, this is then used 
to deal with the case of contraction as outlined above.

For any term~$t$ of labelled \textsc{stl}, we write $|t|$ for
the \textsc{stl}-term obtained by removing all label annotations
and deleting all coercions. Likewise, we write $|X|$ for the 
type obtained by removing all label annotations.
In the following lemma, we also extend the operation~$+$ 
to sequences of label terms: If $r\in \LT^*$ is $L_1\dots L_n$
and $r'\in \LT^*$ is $L_1'\dots L_n'$, then we write
$r+r'$ for $(L_1 + L'_1) \dots (L_n + L'_n)$.
\begin{lemma}
  \label{lem:eta}
  For any \textsc{stl}-type~$X$ there exists labels $a'_1, a'_2, q_1, q_2 \in \LL^*$   
  and labelled terms~$t_1$ and~$t_2$, such that
  $|t_1| = |t_2| = \eta(x,\overline X)$ and 
  such that the judgements
  \[
    \SeqTm{\I x {\overline X[q, a'_1 + a'_2]}}{t_1}{\overline X[q_1, a_1]}
      \qquad\text{and}\qquad
      \SeqTm{\I x {\overline X[q, a'_1 + a'_2]}}{t_2}{\overline X[q_2, a_2]}
  \]
  are derivable for all label terms $q, a_1, a_2\in \LT^*$ for which the types
  ${\overline X[q,  a'_1 + a'_2]}$,
  ${\overline X[q_1, a_1]}$ and ${\overline X[q_2, a_2]}$ are defined.   
\end{lemma}
\begin{proof}
  For any labelled type~$Y$ we define the list of label terms in 
  positive/negative positions in it:
  $P(\NN) = P(\bot) = N(\NN) = N(\bot) = \varepsilon$ (the empty list),
  $N(Y\xto L Z) = L\, P(Y)\, N(Z)$ and
  $P(Y\xto L Z) = N(Y)\, P(Z)$.

  Informally, the $\eta$-expansion of a variable is such that the
  label terms in positive position appear only as annotations of 
  applications.  There are no
  restrictions on the label terms in an application, so an 
  $\eta$-expansion can be typed for arbitrary label terms in positive 
  position.  For the terms in negative position, there are
  constraints however.
  For each label term in a negative position the term contains a 
  $\lambda$-abstraction. Since each
  abstraction must be annotated with a unique label, this 
  leads to the constraint that the label terms in negative position
  can be obtained by coercion from the unique label of
  the abstraction in the term.

  Formally, this can be expressed as follows:
  Let~$X$ be any \textsc{stl}-type and let $X_1$ and $X_2$ be 
  labelled types with $|X_1|=|X_2|=X$.
  Then there exists a term~$t$ with $|t|=\eta(x,X)$ 
  such that $\SeqTm{\I x {X_1}}{t}{X_2}$ is derivable
  whenever
  the list $L_1\dots L_n := P(X_1)N(X_2)$ has the property
  that there are pairwise distinct 
  labels $l_1,\dots, l_n$ such that, for all
  $i\in\{1,\dots,n\}$, the label~$l_i$ is a sub-term of the label term~$L_i$.

  The proof goes by induction on the type~$X$. We spell out the
  case for function types.
  \begin{itemize}
    \item Case $Y\to Z$.
      We have 
      $P(Y_1 \xto {L_1} Z_1)\,N(Y_2 \xto {L_2} Z_2) = 
      N(Y_1)\, P(Z_1)\, L_2\, P(Y_2)\, N(Z_2)$,
      by definition.
      The assumption on this list implies that $P(Y_2)\, N(Y_1)$ has the
      property needed to apply the induction hypothesis to~$Y$.
      Hence, there exists a labelled 
      term~$t_y$ with $|t_y|= \eta(y,Y)$ and $\SeqTm{\I y {Y_2}}{t_y}{Y_1}$.

      Likewise, the list $P(Z_1)\, N(Z_2)$ is such that we can apply
      the induction hypothesis to obtain a term~$t_z$ with $|t_z|= \eta(z,Z)$ 
      and $\SeqTm{\I z {Z_1}}{t_z}{Z_2}$.

      If we define $t' = \lambda^l.\, t_z[\tappl f {L_1} {t_y}/z]$, we
      therefore have $|t'|=\eta(f,Y\to Z)$ and
      $\SeqTm{\I f {Y_1\xto{L_1} Y_1}}{t'}{Y_2\xto{l} Y_2}$.

      By assumption, we know that we can choose~$l$ to be sub-term of~$L_2$, 
      without violating the constraint that each $\lambda$-abstraction
      must be uniquely identified by its label.
      The result therefore follows by applying coercions to~$t'$.
  \end{itemize}
  The assertion now follows as a special case, where only  
  coercions from $a'_1$ to $a'_1+a'_2$ and from $a'_2$ to
   $a'_1+a'_2$ are used.
\end{proof}

With this lemma we can show that each derivation obtained by
CPS-translation can be typed in the labelled variant of~\textsc{stl}. 
\begin{lemma}
  \label{lem:cpsstl}
  If\/ $\SeqTm{\Gamma}{t}{X}$ is derivable in \textsc{stl},
  then there exist label terms~$x^-, \gamma^+ \in \LT^*$, such that,
  for all label terms $x^+, \gamma^- \in \LT^*$  
  for which $\overline \Gamma [\gamma^-, \gamma^+]$ and
  $\overline X [x^-, x^+]$ are defined,
  the sequent
  $\SeqTm{\overline \Gamma [\gamma^-, \gamma^+]}{t'}{\overline X [x^-, x^+]}$
  is derivable for some labelled term $t'$ with $|t'|= \cps t$.
\end{lemma}
\begin{proof}
  The proof goes by induction on the derivation of $\SeqTm{\Gamma}{t}{X}$.
  We consider representative cases.
  To simplify the notation we write just
  $\SeqTm{\overline \Gamma [\gamma^-, \gamma^+]}{\cps t}{\overline X [x^-, x^+]}$
  to express that there exists a labelled term~$t'$ with $|t'|= \cps t$ 
  for which
  $\SeqTm{\overline \Gamma [\gamma^-, \gamma^+]}{t'}{\overline X [x^-, x^+]}$
  is derivable.
  \begin{itemize}
    \item Case \textsc{ax}.
      This case follows directly from Lemma~\ref{lem:eta}.
    \item Case \textsc{$\to$e}.
      By induction hypothesis, there exist label terms $y^-x_1^+,\gamma^+ \in ^*L$ 
      such that
      $
         \SeqTm{\overline\Gamma[\gamma^-,\gamma^+]}{\cps s}{\overline{X \to Y}[y^-x_1^+,y^+x_1^-]}
      $
      is derivable for all label terms $y^+x^-, \gamma^- \in \LT^*$ for which
      all types in the sequent are defined.      
      Also by induction hypothesis, there exists label terms $x_2^-,
      \delta^+\in \LT^*$, such that for all label terms $x_2^+,
      \delta^-$ the sequent
      $\SeqTm{\overline\Delta[\delta^-, \delta^+]}{\cps t}{\overline X[x_2^-, x_2^+]}$
      is derivable.
      In particular, we can choose $x_2^+$ to be $x_1^+$ and 
      $x_1^-$ to be $x_2^+$ and obtain 
      $
      \SeqTm{\overline\Gamma[\gamma^-,\gamma^+],\, \overline\Delta[\delta^-,\delta^+]}{\cps{s\ t}}
      {\overline{Y}[y^-,y^+]}
      $.
      Thus, we have shown that there exist label terms $y^-,\delta^+\gamma^+$, such
      that the sequent is derivable for all label terms $y^+,
      \delta^-\gamma^-$, as was required to show.
    \item Case \textsc{contr}.
      By induction hypothesis, there exist label terms $y^-, \gamma^+, x_1^+, x_2^+$
      such that 
      $\SeqTm{\overline \Gamma[\gamma^-,\gamma^+],\,
        \I {x_1} {\overline X[x_1^-, x_1^+]},\,
      \I {x_2} {\overline X[x_2^-, x_2^+]}}{\cps t}{\overline Y[y^-, y^+]}$
      is derivable for all label terms 
      $y^+, \gamma^-, x_1^-, x_2^-$.
      By Lemma~\ref{lem:eta}, there exist labels $a_1', a_2' \in \LL^*$
      and we can annotate $\eta(x, \overline X)$ to become $t_1$
      and $t_2$ so that
      that $\SeqTm{\I x {\overline X[x^-, a_1' + a_2']}}{t_1}{\overline X[x_1^-, x_1^+]}$
      and
      $\SeqTm{\I x {\overline X[x^-, a_1' + a_2']}}{t_2}{\overline X[x_2^-, x_2^+]}$
      are derivable.

      We annotate the two copies of $\eta(x, \overline X)$
      in the CPS-translation of contraction as $t_1$ and $t_2$.

      Overall we obtain that there exist label terms $y^-, \gamma^+,
      a_1'+a_2'$, such that 
      \[
        \SeqTm{\overline \Gamma[\gamma^-,\gamma^+],\,
        \I {x} {\overline X[x^-, a_1' + a_2']}}
        {\cps{t[x/x_1,x/x_2]}}{\overline Y[y^-, y^+]}
      \]
      is derivable for all $y^+, \gamma^-, x^-$, 
      which shows the assertion.\qedhere
  \end{itemize}
\end{proof}

\noindent This lemma justifies the definition of $\kw{CpsDefun}$ also
for \textsc{stl}: Given a derivation of $\SeqTm{\Gamma}{t}{X}$,
annotate its CPS-translation using the lemma, so that
$\SeqTm{\Gamma[\gamma^-, \gamma^+]}{\cps t}{X[x^-,x^+]}$
becomes derivable, and take
$\kw{CpsDefun}(\SeqTm{\Gamma}{t}{X}) 
:=(x^-\gamma^+, D(\cps t), x^+\gamma^-)$.

\subsection{Interactive Interpretation.}
We now show how the Int-translation can be extended to
\textsc{stl} and how it relates to defunctionalization.
To this end, we again consider a variant of the type system with 
subexponential annotations.
We extend \linexp to \stlexp 
by adding subexponential annotations to the contraction rule
and by adding a new rule \R{struct} for weakening of subexponential annotations.
Rule \R{struct} makes the type system more well-behaved and also
increases the expressive power of the system. It is needed 
at the end of this Section in the 
proof of Proposition~\ref{prop:stl}.

The new rules of \stlexp are shown in Figure~\ref{fig:stlexp}.
\begin{figure}
\begin{prooftree}
  \AxiomC{$ \SeqU{\Gamma,\, \J y A X,\, \J z B X}{t}{Y} $}
  \LeftLabelSc{contr}
  \UnaryInfC{$ \SeqU{\Gamma,\, \J x {(A+B)} X}{t[x/y,x/z]}{Y} $}
\end{prooftree}      
\begin{prooftree}
  \AxiomC{$ \SeqU{\Gamma,\, \J x A X}{t}{Y} $}
  \RightLabel{$A \sleq B$}
  \LeftLabelSc{struct}
  \UnaryInfC{$ \SeqU{\Gamma,\, \J x B X}{t}{Y} $}
\end{prooftree}
\caption{Additional Rules of \stlexp over \linexp}
\label{fig:stlexp}
\end{figure}
To understand the annotations on \R{contr},
recall the explanation of subexponentials as 
making explicit the environment in which a variable is being used. The
judgement in the premise of \R{contr} tells us that
the variables~$y$ and~$z$
are used in environments with additional values of type~$A$ and~$B$
respectively. The subexponential $A+B$ in the conclusion 
tells us that~$x$ may be used in two ways: first in an environment 
that contains an additional variable of type~$A$ and second in one with an additional 
variable of type~$B$. The coproduct identifies the 
two copies of~$x$.
Rule \R{struct} has a side condition $A\lhd B$, which expresses 
that~$A$ is a retract of~$B$, i.e.~that any value of type~$A$ can be
encoded into one of type~$B$. Formally, $A\lhd B$ holds if and only if there
exist target expressions $\SeqTm{\I x A}{s}{B}$ and $\SeqTm{\I y B}{r}{A}$,
such that $r[s[v/x]/y]\red^* v$ holds for any target value~$v$ of type~$A$.
Notice in particular, that for isomorphic types $A\iso B$, we have
both $A\lhd B$ and $B\lhd A$.

The Int-translation of the new rules of \stlexp 
is shown in Figure~\ref{fig:stlexpint}.
Rule \R{contr} is interpreted by use of the isomorphism
$(A+B)\times C \iso A\times C + B \times C$, which is 
implemented using case distinction. 
A message of type $(A+B)\times X^+$ has the form $<\inl(a),x>$ 
or $<\inr(b),x>$. Depending on the case, the message is forwarded
to the occurrence of either~$y$ or~$z$.
In the interpretation of rule \R{struct}, one chooses~$s$ and~$r$ 
to witness $A\lhd B$ as defined above.
The interpretation will be sound for any such choice of~$s$ and~$r$,
see.
\begin{figure}
\begin{center}
  \includegraphics{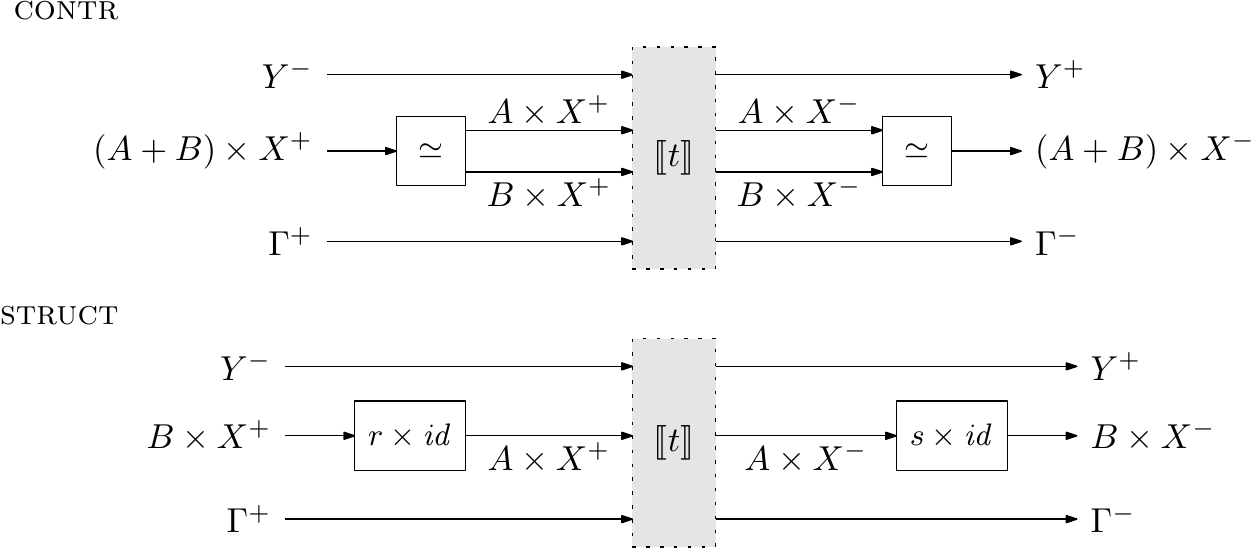}
\end{center}
\caption{Int-Interpretation of new rules in \stlexp}
\label{fig:stlexpint}
\end{figure}

\subsection{Relating the Translations}
We have now defined two translations from \stlexp to the target language.
To relate them, we begin by spelling out a simple example to illustrate that both
translations treat contraction in the same way. 

\begin{example}
Consider again the source term 
$\tlami x \NN {x + x}$.
Its Int-interpretation may be depicted as follows.
\begin{center}
  \includegraphics{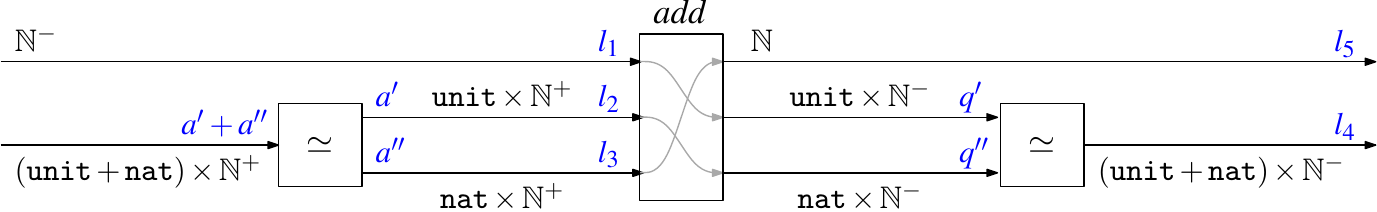}
\end{center}
The box labelled $\mathit{add}$ is defined as in the Introduction, up to
uses of the isomorphism $\unit \times \NN^+ \iso \NN^+$.
The interpretation of rule \R{contr} inserts the two boxes labelled~$\simeq$,
which denote the canonical isomorphism of their type.

This program implements the term $\tlami x \NN {x + x}$ as follows: 
To compute the result of the function when applied to the actual argument 42, one
connects the output of type $(\unit+\VN)\times \NN^-$ to the input of type
$(\unit+\VN)\times \NN^+$ such that when the value $<k',<>>$ arrives at
the output port, then the value $<k',42>$ is fed back to the input
port. 

Consider now the CPS-translation of the term
$\tlami x \NN {x + x}$.
If we omit the $\eta$-expansions in the translation of 
variables for simplicity, then
we obtain the term
\begin{equation}
  \label{eq:1}
  \lambda^{l_1} <x,k>.\, \tappl {t_1} {q'} {(\lambda^{l_2} m.\, \tappl {t_2} {q''} {(\lambda^{l_3} n.\, \tappl k {l_5} {(m+n)})})}
\end{equation}
of type 
$ (((\NN \xto{a'+a''} \bot) \xto{q'} \bot) 
  \times (\NN \xto{l_5} \bot))
  \xto{l_1} \bot$,
wherein
$t_1$ and $t_2$ are
the $\eta$-expansions
$(\lambda^{q'} k.\, \tappl x {l_4} {\coercl{a'+a''}{\lambda^{a'} n.\, \tappl k {l_2} n)}}$
and
$(\lambda^{q''} k.\, \tappl x {l_4} {\coercr{a'+a''}{\lambda^{a''} n.\, \tappl k {l_3} n)}}$
respectively.
These two $\eta$-expansions come from the CPS-translation of contraction.
Defunctionalization of the term in (\ref{eq:1}) leads to the equations
\begin{align*}
  \Capply_{l_1}(<>, <x,k>) &= \Capply_{q'}(q'(x), {l_2}(x,k))
  \\
  \Capply_{l_2}(l_2(x,k), m) &= \Capply_{q''}(q''(x), l_3(m,k))
  \\
  \Capply_{l_3}(l_3(m,k),n) &= \Capply_{l_5}(k, m+n),
\end{align*}
and the subterms~$t_1$ and~$t_2$ add the following equations:
\begin{equation}
\begin{aligned}
  \label{eq:eta}
  \Capply_{q'}(q'(x), k) &= \Capply_{l_4}(x, \kw{inl}(a'(k)))
  &
  \Capply_{a'}(a'(k), n) &= \Capply_{l_2}(k, n)
  \\
  \Capply_{q''}(q''(x), k) &= \Capply_{l_4}(x, \kw{inr}(a''(k)))
  &
  \Capply_{a''}(a''(k), n) &= \Capply_{l_3}(k, n)
  \\
  \Capply_{a'+a''}(k, n)  
  &=
  \tmcase k {f_1} {\Capply_{a'}(f_1, n)} 
            {f_2} {\Capply_{a''}(f_2, n)} 
\end{aligned}
\end{equation}
In order to understand how this program works, it is perhaps again
useful to apply the above term to the argument
$<\lambda^{l_4} k.\, \tappl k {a' + a''} 42,\, \lambda^{l_5} n.\, \texttt{print\_int}(n)>$.
Defunctionalization then yields two additional equations.
\begin{align*}
  \Capply_{l_4}(<>, k) 
  &=  \Capply_{a' + a''}(k, 42)
  \\
  \Capply_{l_5}(<>, n) &= \texttt{print\_int}(n)
\end{align*}
Note how this program computes the result in the same way as 
the one obtained by Int-interpretation above.
Both programs have the same skeleton.
The points corresponding to the $\Capply$-equations are
labelled in the Int-interpretation above.
Notice in particular how the equations~(\ref{eq:eta}) that come from the 
$\eta$-expansion in the CPS-translation of contraction correspond
to the isomorphisms added by the interpretation of 
rule \R{contr} in the Int-interpretation.
\end{example}

\begin{proposition}
  \label{prop:skeletonstl}
  For any\/ $\SeqTm{\Gamma}{t}{X}$ derivable in \stlexp,
  there exists a target program\/ $\semc{\SeqTm{\Gamma}{t}{X}}$ that is a representative 
  of the Int-interpretation of the derivation of the sequent and
  that has the same skeleton as\/ $\kw{CpsDefun}(\SeqTm{\Gamma}{t}{X})$.
\end{proposition}
\begin{proof}
  The proof goes by induction on the derivation, just as for
  Proposition~\ref{prop:skeleton}.
  The only new case is that for contraction.
  To handle this case, consider the defunctionalization
  of the two $\eta$-expansions
  $\SeqTm{\I x {\overline X[q, a'_1 + a'_2]}}{t_1}{\overline X[q_1, a_1]}$
  and
  $\SeqTm{\I x {\overline X[q, a'_1 + a'_2]}}{t_2}{\overline X[q_2, a_2]}$
  from Lemma~\ref{lem:eta}.
  Let us write $q_1(i)$ for the $i$-th label term in the list~$q_1$,
  and likewise for the other lists.
  Observe that the defunctionalization of the terms~$t_1$ and~$t_2$
  yield equations of the following shape
  for all possible indices~$i$:
  \begin{align*}
    \Capply_{q_1(i)}(-, -) &= \Capply_{q(i)}(-, \kw{inl}(-))\\
    \Capply_{q_2(i)}(-, -) &= \Capply_{q(i)}(-, \kw{inr}(-))\\
    \Capply_{(a_1' + a_2')(i)}(x, -) &= \tmcase x y {\Capply_{a_1(i)}(y, -)} z
   {\Capply_{a_2(i)}(z, -)}\\
   \Capply_{a_1'(i)}(-, -) &= \Capply_{a_1(i)}(-, -)
  \end{align*}

  These equations have the same skeleton as an appropriate choice of
  equations for the Int-interpretation of rule \R{contr}.
  One can thus choose a representative of the Int-interpretation 
  having the same skeleton as the program obtained from CPS-translation
  and defunctionalization.
\end{proof}

The proposition establishes a relation of the Int-interpretation and
CPS-translation followed by defunctionalization for terms typeable in \stlexp.
An obvious question is how much of \textsc{stl} is covered by this result.
In the rest of this section we show that with rule \R{struct} and 
recursive types in the target language, in fact any \textsc{stl}-term 
can be typed in \stlexp.

We first give an example to show how recursive target types appear in 
the two translations.
When we discussed defunctionalization for \textsc{stl}, we have
already remarked that recursive types are needed in the target
language to treat the full simply-typed $\lambda$-calculus.
The following example (i)~illustrates the use of 
\R{struct} and recursive types in \stlexp; and (ii)~shows that
recursive types may appear even in the defunctionalization of the
simply-typed $\lambda$-calculus.

\begin{example}
  \label{ex:rec}
An example that illustrates why without recursive types in the target 
language not every \textsc{stl}-term would be typeable in \stlexp
is the application $t\ s$, where
$t=\lambda g.\, g\ (\lambda x.\, g\ (\lambda y.x))$
and
$s=\lambda f.\, f\ (f\ (\lambda x.x))$.
The terms~$t$ and~$s$ can be
given types
$(\unit +\alpha)\cdot (\alpha\cdot(\alpha\cdot X \lollipop X) \lollipop X) \lollipop X$
and
$(\unit+\beta)\cdot (\beta\cdot X \lollipop X) \lollipop X$
respectively, for a certain type~$X$.
In these types the subexponential annotations have been simplified
using only the isomorphism $(-)\times \unit \iso (-)$, which can be used by
means of rule \R{struct}.

Without recursive types in the target language, the application $t\ s$
could not be typed, as this would require us to unify
$(\unit +\beta)\cdot (\beta\cdot X \lollipop X) \lollipop X$
with
$\alpha\cdot(\alpha\cdot X \lollipop X) \lollipop X$,
which would require unifying $\beta$ and $\unit +\beta$. 
With recursive types, however, we can simply
let $B := \mu \beta.\,\unit +\beta$ and instantiate the type
variable~$\beta$ to be~$B$.
Since we have $\unit + B \lhd B$, we can 
use rule \R{struct} to give~$s$ the type
$B\cdot (B\cdot X \lollipop X) \lollipop X$
and with this give $t\ s$ the type~$X$ .

It is interesting to note that $\kw{CpsDefun}$ maps $t\ s$ to a 
program that also uses recursive types.
An annotation of the CPS-translation of $t\ s$ in the labelled version
of \textsc{stl} is:
\begin{align*}
    \cps t &= \lambda^{l_1}<g,k>.\, 
      \tappl g {l_4} 
      {<\coercl{}{\lambda^{l_2}<x,k_1>.\, \tappl g {l_4} {<\coercr{}{\lambda^{l_3}<y, k_3>.\, \tappl x {l_5+l_6} {k_3}}, k_1>}}, k>}
  \\
  \cps s &= \lambda^{l_4} <f,k>.\, 
     \tappl f {l_2+l_3} 
     { <\coercl{}{\lambda^{l_5} k_2 .\, \tappl f {l_2+l_3} {<\coercr{}{\lambda^{l_6} <x, k_1>.\, \tappl x {l_8} {k_1}}, k_2>}}, 
           k> }           
  \\
  \cps{t\ s} &= \lambda^{l_7} k.\, \tappl{\cps t} {l_1} {<\cps s, k>}
\end{align*}
The types~$\tau_{(-)}$ that appear in the defunctionalization are:
\begin{align*}
  \tau_{l_1} &= \kw{datatype}\ l_1\ \kw{of}\ \unit 
  &
  \tau_{l_2} &= \kw{datatype}\ l_2\ \kw{of}\ \tau_{l_4} \\
  \tau_{l_3} &= \kw{datatype}\ l_3\ \kw{of}\ (\tau_{l_5} + \tau_{l_6}) 
  &
  \tau_{l_4} &= \kw{datatype}\ l_4\ \kw{of}\ \unit \\
  \tau_{l_5} &= \kw{datatype}\ l_5\ \kw{of}\ (\tau_{l_2} + \tau_{l_3})  
  &
  \tau_{l_6} &= \kw{datatype}\ l_6\ \kw{of}\ \unit
\end{align*}
The types $\tau_{l_3}$ and $\tau_{l_5}$ are mutually recursive.

The reader familiar with Game Semantics may recognise the term~$t$ as 
one of the Kierstead terms of order three that is often used to 
illustrate the need for
justification pointers in Hyland-Ong-games. The other Kierstead term
of order three 
$t' = \lambda g.\, g\ (\lambda x.\, g\ (\lambda y.y))$ can be given type
$(\unit +\alpha)\cdot (\alpha\cdot(\kw{unit}\cdot X \lollipop X) \lollipop X)
\lollipop X$. With this term it \emph{is} possible to give a type to $t'\ s$ 
without recursive types by setting
$\alpha := \kw{unit} + \kw{unit}$ and $\beta := \kw{unit}$.
\end{example}

We end this section by showing that with rule \R{struct} and recursive
types in the target language, \stlexp can indeed type any
\textsc{stl}-term.
Suppose $\SeqTm{\Gamma}{t}{X}$ is a typing judgment of
\stlexp.
Write $|X|$ and $|\Gamma|$ for the type and
context of \textsc{stl} obtained by removing all subexponential 
annotations, i.e.~replacing any $\lolli A Y Z$ with $Y\to Z$ and 
removing subexponentials in the context.
With this notation we have:
\begin{proposition}
  \label{prop:stl}
  If\/ $\SeqTm{\Gamma}{t}{X}$ is derivable in \textsc{stl},
  then there exist $\Delta$ and $Y$ with
  $\Gamma = |\Delta|$ and $X=|Y|$, such that
  $\SeqTm{\Delta}{t}{Y}$ is derivable in \stlexp.
\end{proposition}
\begin{proof}
  Using rule \R{struct}, the following rules are derivable.
    \begin{center}
      \begin{tabular}{cc}
        \AxiomC{\phantom{X}}
        \LeftLabelSc{ax}
        \RightLabel{$\unit \lhd \alpha_1$}
        \UnaryInfC{$ \SeqU{\J x {\alpha_1} X}{x}{X} $}
        \bottomAlignProof
        \DisplayProof
      \end{tabular}
    \end{center}
    \begin{center}
      \begin{tabular}{cc}
        \AxiomC{$ 
        \SeqU{\Gamma}{s}{\lolli A X Y}
        $}
        \AxiomC{$ 
        \SeqU{\J {x_1} {A_1} {X_1},\dots,  \J {x_n} {A_n} {X_n}}{t}{X}
        $}
        \LeftLabelSc{$\lollipop$e}
        \RightLabel{$ 
            \begin{array}{l} 
                A\times A_1\lhd\alpha_1,\dots,\\
                A\times A_n\lhd\alpha_n
             \end{array}$}
        \BinaryInfC{$ \SeqU{\Gamma,\, \J {x_1} {\alpha_1} {X_1},\dots,  \J {x_n} {\alpha_n} {X_n}}{s\ t}{Y} $}
        \bottomAlignProof
        \DisplayProof
      \end{tabular}
    \end{center}
    \begin{center}
      \begin{tabular}{cc}
        \AxiomC{$ \SeqU{\Gamma }{s}{\NN} $}
        \AxiomC{$ \SeqU{\J {x_1} {A_1} {X_1},\dots,  \J {x_n} {A_n} {X_n}}{t}{\NN} $}
        \LeftLabelSc{add}
        \RightLabel{$ 
            \begin{array}{l} 
                \VN\times A_1\lhd\alpha_1,\dots,\\
                \VN\times A_n\lhd\alpha_n
             \end{array}$}
        \BinaryInfC{$ \SeqU{\Gamma,\, \J {x_1} {\alpha_1} {X_1},\dots,  \J {x_n} {\alpha_n} {X_n}}{s+t}{\NN} $}
        \bottomAlignProof
        \DisplayProof
      \end{tabular}
    \end{center}
    \begin{center}
      \begin{tabular}{cc}
        \AxiomC{$ \SeqU{\Gamma,\, \J y A X,\, \J z B X}{t}{Z} $}
        \LeftLabelSc{contr}
        \bottomAlignProof
        \RightLabel{$ (A+B)\lhd \alpha_1 $}
        \UnaryInfC{$ \SeqU{\Gamma,\, \J {x} {\alpha_1} {X}}{t[x/y, x/z]}{Z} $}
        \bottomAlignProof
        \DisplayProof
      \end{tabular}
    \end{center}
    We only need \R{struct} to derive these rules.

  If we use these derived rules with fresh target  type variables for the~$\alpha_i$ 
  and disregard the $\lhd$-side-conditions for
  now, then together with the 
  unchanged rules \R{weak}, \R{exch} \R{$\to$i}, \R{if},
  \R{num}, we can construct a skeleton of a typing derivation
  for~$t$. This exists because~$t$ is typeable in \textsc{stl}.

  To make this into a proper \stlexp type derivation, it just remains to solve
  all the $\lhd$-constraints. 
  The constraints all have the form $A\lhd \alpha$, i.e.~the
  right-hand side of any constraint is a type variable.
  With recursive types, it is easy to solve such constraints:
  Let $A_1\lhd \alpha,\dots,A_n\lhd \alpha$ be all constraints
  with~$\alpha$ on the right-hand side. A solution for it is
  $\alpha := \mu \alpha.\, A_1 + \dots + A_n$. In this way, we can solve the
  constraints for the type variables one after the other and so obtain
  a correct typing derivation.
\end{proof}
We note that the proof provides a simple type inference procedure for
\stlexp. 
It is adapted from the simple type inference algorithm in~\cite{aplas10}.
Since~\cite{aplas10} is concerned with \textsc{logspace}-computation, 
recursive types are not allowed there, and the constraints
are solved by trying to unify $\alpha$ with $ A_1 + \dots + A_n$ instead of
setting $\alpha := \mu \alpha.\, A_1 + \dots + A_n$. 
This simple heuristic does not work for all \textsc{stl} terms and
we need to allow recursive types to prove the above proposition in general.

\section{Recursion}
\label{sect:fix}

We conclude by explaining how the Int-interpretation and the
subexponential annotations can be extended to handle 
the fixed point combinator of PCF.

Subexponential annotations for the fixed-point combinator can be given
by
\begin{prooftree}
  \AxiomC{\phantom{X}}
  \LeftLabelSc{fix}
  \UnaryInfC{
      $\kw{fix}_{X}\colon (\tlist A) \cdot (\lolli A X X) \to X$
    }
\end{prooftree}
where $\tlist A$ abbreviates $\mu \alpha.\, \unit +
A \times \alpha$. 

The Int-interpretation of this term can be defined as follows:
\begin{center}
  \vspace{1em}

  \includegraphics{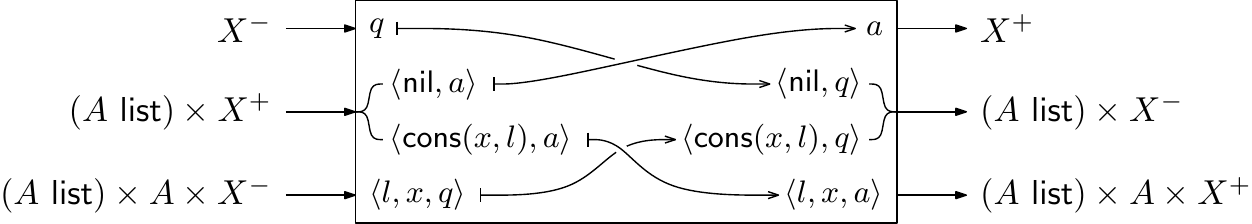}

  \vspace{1em}
\end{center}
Here, lists are used to implement a call stack.
The function that we take the fixed point of has type $A\cdot X \lollipop X$.
This tells us that it needs to store a value of type~$A$ whenever it 
requests its argument. Thus, an activation record should be a stack of values
of type~$A$, which we encode as a list.

The appearance of the type $(\tlist A)$ can also be explained from
the subexponential type system.
Clearly, the fixpoint combinator should have type
$\kw{fix}\colon B \cdot (A\cdot X \lollipop X) \lollipop X$
for some~$B$. It should be defined to satisfy the equation
$\kw{fix}\ f = f\ (\kw{fix}\ f)$.
Consider typing judgements for the two terms in this equation.
For the left-hand term we have
$\SeqU{\J f {B} {(A\cdot X \lollipop X)}}{\kw{fix}\ f}{X}$.
As~$f$ appears twice in the right-hand term, we must use contraction
to type it:
$\SeqU{\J f {(1+A\times B)} {(A\cdot X \lollipop X)}}
{f\ (\kw{fix}\ f)}{X}$.
Notice now that we can give both terms the same type (using
\R{struct}) if we can solve the type equation $B\iso 1+A\times B$.
We are thus naturally lead to choosing $B := \tlist A$.
That case distinction appears in the above implementation of the
fixed-point combinator is due to the duplicated use of the variable~$f$. 

The Int-interpretation implements recursion in a similar way as 
CPS-translation and defunctionalization. The CPS-translation of
the fixed point combinator is
\[
  \lambda^{q_1} <f,k>.\, 
  \kw{fix}_{\overline X}\ (\lambda^{q_4} g.\, \lambda k_1.\,f\ <\lambda^{q_3} k_2.\,g\ k_2,\, \lambda^{q_2} x.\, k_1\ x>)\ k
  \enspace.
\]
A possible defunctionalization (without using control flow information)
of this term is:
\begin{align*}
  \Capply(q_1(), <f,k>) &= \Capply(q_4(f), k) 
  &
  \Capply(q_2(k_1), x) &= apply(k_1, x)
  \\
  \Capply(q_3(g), k_2) &= apply(g, k_2)
  &
  \Capply(q_4(f), k_1) &= \Capply(f, <q_3(q_4(f)), q_2(k_1)>)
\end{align*}
Informally, the first three definitions to correspond to the inputs of
the Int-interpretation above. A call to $\Capply(q_1(), <f,k>)$ starts
the recursion, a call to $\Capply(q_2(k_1), x)$ corresponds to
the step function (that the fixed point is taken of) returning a result,
and a call to $\Capply(q_3(g), k_2)$ corresponds to the step function
requesting its argument. The final equation does not contribute to the
external interface of the program and is used to implement the fixed point.
The call stack, which above is encoded using lists, appears more
implicitly in the continuations here.

\section{Conclusion}

We have observed that the non-standard compilation methods based on
computation by interaction are closely related to CPS-translation and
defunctionalization.  The interpretation in an interactive model may
be regarded as a simple direct description of the combination of
CPS-translation, defunctionalization and a final optimisation of
arguments.  It may be seen as a 
simple nameless formulation of a combined CPS-translation and
defunctionalization and it provides an alternative way of encoding 
continuations.

We have seen in this paper that working out the
technical details of defunctionalization with explicit labels can
become quite technical. The interactive interpretation admits 
a high-level description that abstracts from implementation details.
Interactive model constructions may perhaps be
useful in simplifying uses of defunctionalization.
In the other direction,
being aware that interactive models are related to standard
compilation methods may help to improve non-standard compilation
schemes based on interactive methods, such as~\cite{Ghica07,intml}. We
may hope that some of the many existing techniques for compiler
implementation can be adapted usefully to the non-standard schemes.

Types with subexponential annotations, in this form originally introduced in IntML~\cite{intml}, 
provide a logical account for the issues of managing value
environments that are inherent to defunctionalization.
With subexponential annotations, the type of a higher-order term fully 
specifies the interface of the target program obtained from it.
The type system contains enough information in order to give a fully
compositional definition of the translation to the target language.

The subexponential type system makes explicit issues 
that appear with defunctionalization and separate compilation. 
For example, in order to suppose we want to compile a function $f\colon \lolli A X Y$
separately from the rest of the program. Then one may compile
the main program $\SeqTm{\J f B {(\lolli A X Y)}}{t}{Z}$ 
and $\SeqTm{}{f}{\lolli A X Y}$ separately. A linker can combine the
resulting two programs knowing only their types.
Of course, the problems associated with defunctionalization and separate
compilation do not just disappear. 
Suppose, for example, we only know the term~$f$, but not the
program~$t$ in which it will be used.
Suppose further that $X$ has the form $\lolli C \NN \NN$. Then 
the choice of the subexponential annotation~$C$ will limit 
which arguments~$f$ can be applied to; $f$ can only be applied to arguments
of type $\lolli D \NN \NN$ with $D \lhd C$. Choosing~$C$ without
knowledge of~$t$ is possible, for example, if 
the target language has a type $\kw{Heap}$ 
with $D\lhd \kw{Heap}$ for any type~$D$ (anything can be stored on the
heap). Then one may simply choose $C$ to be $\kw{Heap}$. 
This is not the only possible choice; for performance reasons
one may consider performing a more precise analysis in a linker or similar.
The point is that the subexponential type system allows us to express
such issues at a high level and to apply different possible solutions.
Another example for this point is the explanation of the appearance
of recursive target types in Section~\ref{sect:stl},
which was given in terms of subexponential type annotations.

Subexponentials refine the exponentials in AJM
games~\cite{AbramskyJM00}, where $!X$ is implemented using
$\omega\cdot X$, where~$\omega$ is a type of unbounded natural numbers.
If we had used full exponentials in the Int-interpretation 
above, then we would have obtained a compilation method that encodes
function values as values in~$\omega$, which is akin to storing closures on the
heap. Subexponentials give us more control to avoid such encodings where
unnecessary.
The subexponential type system in this paper has
its origin in Bounded Linear Logic~\cite{girardscedrovscott,Schopp07}.
It is also similar
to the type system for Syntactic Control of Concurrency (SCC)~\cite{DBLP:journals/tcs/GhicaMO06}.
A main
difference appears to be that while SCC controls
the number of program threads,
subexponentials account for both
the threads and their local data.

The observation that there is a connection between 
game models and continuations is not new. It appears, for example,
in Levy's work on a jump-with-argument calculus~\cite{cbpv}
and in Melli{\`e}s work on tensorial logic~\cite{mellies12}.
Connections of game models to compilation have also been made, e.g.~\cite{mellies}.
Furthermore, it is well-known that continuation passing is related to 
message passing, see e.g.~\cite{thieleckePhd}.
However, we 
are not aware of work that makes explicit a connection to defunctionalization.

We believe that the connection between game models and machine languages 
deserves to be better known and studied further.
The call traces in this paper, for example, should 
have the same status as plays in Game Semantics.
This suggests that techniques from
Game Semantics help to analyse the possible traces 
of compiled machine code.
For Java-like languages, there is recent work that connects 
fully abstract trace semantics~\cite{DBLP:conf/esop/JeffreyR05} and
Game Semantic models~\cite{DBLP:conf/popl/MurawskiT14}.
We hope that similar connections can be identified for the traces of 
machine code generated by compilers.

Work on concurrency and process calculus emphasises the interactive nature of 
computation. 
Milner's translation of the $\lambda$-calculus in the 
$\pi$-calculus \cite{DBLP:journals/mscs/Milner92} can be understood
as a CPS-translation~\cite{DBLP:journals/mscs/Sangiorgi99}. 
This connection was made more than once: see
\cite[\S10]{DBLP:journals/mscs/Sangiorgi99} for historical notes.
The interactive model that we have studied in this paper can be seen
as a very static form of communicating processes,
without process mobility or channel reconfiguration.
In Milner's translation, channel names identify continuations and these
are passed around explicitly. We have seen in this paper that 
by taking into account control flow information, it is possible
to avoid passing continuations names, as they can be determined statically.
It is an interesting direction for further work to find out if
Milner's translation relates to CPS-translation and defunctionalization
without control flow information. 
In this direction, it may perhaps be possible to connect to 
the interesting results of Berger~et~al.~\cite{BergerHY01}.

In further work, we should like to understand possible connections 
to Danvy's work on defunctionalized interpreters~\cite{danvy},
or more generally to work on continuations and abstract machines, 
e.g.~\cite{DanosHerbelinRegnier,DBLP:journals/jfp/StreicherR98}. 
A relation is not obvious: Danvy considers the 
defunctionalization of particular implementations of interpreters,
while here we show that the whole compilation itself may be
described extensionally by the Int construction. 

In another direction, an interactive view of CPS-translation and
defunctionalization may also help in identifying mathematical structure
of efficient compilation methods. In particular, capturing call-by-value 
defunctionalizing compilation, perhaps similar to~\cite{DBLP:conf/esop/CejtinJW00},
should be interesting. Other interesting issues are 
efficient separate compilation and polymorphism: 
the interpretation in $\Int\TT$ is compositional
and polymorphism can also be accounted for~\cite{Schopp07}.

Finally, this paper clarifies the definition of IntML~\cite{intml}, which
was introduced to capture \textsc{logspace}. In~\cite{intml} we 
observed that IntML supports control operators, such as~$\kw{callcc}$,
but their status remained somewhat unclear. It now turns out that
the $\kw{callcc}$ combinator of~\cite{intml} 
may be understood as the defunctionalization of a standard
CPS implementation of $\kw{callcc}$. 

\paragraph{Acknowledgements.} 
I would like to thank the anonymous referees for their constructive
feedback and suggestions, which helped to improve the presentation of the results.

\bibliographystyle{plain}
\bibliography{defun}
\end{document}